\documentclass[reqno]{amsart}
\usepackage{amsmath,amssymb}
\usepackage{graphics,epsfig,color}
\usepackage{booktabs,siunitx}  
\usepackage[mathscr]{euscript}
\usepackage{verbatim}
\usepackage{hyperref,color}
\usepackage{wasysym}
\usepackage[dvipsnames]{xcolor}
\usepackage[framemethod=default]{mdframed}
\usepackage[Lenny]{fncychap}

\usepackage{tikz}

\usepackage[all]{xy}
\xyoption{all}

\theoremstyle{plain}
\newtheorem{theorem}{Theorem}[section]

\theoremstyle{definition}

\theoremstyle{remark}
\newtheorem{remark}[theorem]{Remark}

\newcommand{\bb}[1]{{\mathbb #1}}

\renewcommand{\phi}{\varphi}

\renewcommand{\tilde}{\widetilde}
\renewcommand{\hat}{\widehat}

\definecolor{light}{gray}{.9}

\title[Gibbsian stationary non equilibrium states]{Gibbsian stationary non equilibrium states
}
\author[L.\ De Carlo]{Leonardo De Carlo}
\address{Leonardo De Carlo \hfill\break \indent
  GSSI
  \hfill\break\indent
  Viale Francesco Crispi 7,   67100  L'Aquila, Italy}
\email{neoleodeo@gmail.com}

\author[D.\ Gabrielli]{Davide Gabrielli}
\address{Davide Gabrielli \hfill\break \indent
  DISIM, Universit\`a dell'Aquila
  \hfill\break\indent
  Via Vetoio,   67100 Coppito, L'Aquila, Italy}
\email{gabriell@univaq.it}

\begin{document}

\begin{abstract}
We study the structure of stationary non equilibrium states for interacting particle systems from a microscopic viewpoint.
In particular we discuss two different discrete geometric constructions. We apply
both of them to determine non reversible transition rates corresponding to a fixed invariant measure.
The first one uses the equivalence of this problem with the construction of divergence
free flows on the transition graph. Since divergence free flows are characterized by cyclic decompositions we
can generate families of models from elementary cycles on the configuration space.
The second construction is a functional discrete Hodge decomposition for
translational covariant discrete vector fields. According to this, for example, the instantaneous current
of any interacting particle system on a finite torus can be canonically decomposed in a gradient part, a circulation term
and an harmonic component. All the three components are associated with functions on the configuration space.
This decomposition is unique and constructive. The stationary condition can be interpreted
as an orthogonality condition with respect to an harmonic discrete vector field and we use this decomposition to construct models having a fixed invariant measure.

\bigskip


\smallskip

\end{abstract}

\maketitle
\thispagestyle{empty}




\section{Introduction}
We propose and study two different ways to generate non reversible stochastic particle systems having a fixed Gibbsian invariant measure. The two methods are quite general and correspond to two different discrete geometric constructions that have an interest in themselves.

The construction of reversible Markovian dynamics having a fixed target invariant measure is a classic problem that is at the core  of the Monte Carlo method. This is very useful in several applicative problems and the literature concerning the method is huge. We suggest for example \cite{D,Metr} and the references therein.

The construction of non reversible Markovian dynamics is more difficult. This is an interesting problem since the violation of detailed balance is expected for example to speed up the convergence of the dynamics \cite{Bi,KJZ,RBS,SB}. Another motivation come from the fact that the understanding of collective and macroscopic behavior of interacting multi-component systems is based on the knowledge of the stationary state. For example the knowledge of the invariant measure is a necessary requirement
in order to have a direct connection between the microscopic details of the dynamics and the transport coefficients. We are constructing a large class of models for which this is possible. For example the violation of Einstein relation that is expected for anisotropic
systems \cite{BBC} could be observed directly inside our class of models.
Several references are dedicated to the problem of finding non reversible Markov chains having a fixed target invariant measure in specific models, see for example \cite{BB,CT,FGS,G,GL,LG,RBS,PSS}. We study the problem in the case of interacting particle systems and describe two somehow general approaches. We consider for simplicity the case of stochastic lattice gases but the methods can be generalized to more general models.

Our first approach is based on the equivalence between the problem of finding rates that have a fixed invariant measure and the problem of finding a divergence free flow on the transition graph. A flow is a map that associates a positive number, representing the mass flowed, to each oriented edge of a graph. The divergence of a flow is the amount of mass flowing outside of a vertex minus the amount of mass flowing into the vertex.   The usefulness of this equivalence is on the fact that the structure of divergence free flows is quite well understood. Any divergence free flow can be indeed represented as a superposition of  flows associated with elementary cycles \cite{GO} (see also \cite{BFG,GV,MC} for more detailed results and applications and section 2 for precise definitions). A local non reversible dynamics is then generated considering elementary cycles on the configuration space for which the configuration of particles is frozen outside of a bounded region. The construction of divergence free discrete vector fields and flows on graphs and oriented graphs using cyclic decompositions is a classic one \cite{GO,B}. There is a large number of applications of this idea in probability and statistical mechanics. For example in \cite{MC,Bi,RBS} it is used to construct non reversible Markov models. This approach can be seen as a counterpart of the combinatorial representation of the invariant measures of Markov chains
in terms of oriented spanning trees \cite{FW}. The novelty of our results is on the fact that we apply it to particle systems. In particular we have a large number of cycles on the big configuration space and we have to suitably assign the weights in such a way to
have local transition rates. In all our models there is a non trivial interplay between the combinatorial structure
of the configuration space and the structure of the physical lattice where the particles evolve.

A very interesting issue here is to understand the relation between the speed of convergence
towards the invariant measure and the structure of the cyclic decomposition of the flow. A natural conjecture is to have a faster convergence in presence of global cycles.

The second approach is based on a new functional Hodge decomposition for translational covariant discrete vector fields. The discrete Hodge decomposition is the splitting of any discrete vector field into a gradient component, a circulation and an harmonic part. Consider a discrete vector field depending in a translational covariant way on  configurations of particles. We show that it is possible to do a similar splitting where the gradient component and the circulation term are respectively a gradient and an orthogonal gradient (in two dimensions) with respect to the shift of two functions on configurations.

A motivation for the functional Hodge decomposition comes from the theory of hydrodynamic scaling limits. In the case of diffusive systems the derivation of the hydrodynamic scaling limit is much simplified in the case of models having the instantaneous current of gradient type. When also the invariant measure is known it is possible to have an explicit form of the transport coefficients.
In the non gradient cases the theory develops starting from an orthogonal decomposition \cite{VY} and the transport coefficients have a representation in terms of the Green Kubo formulas \cite{KL99,Spohn}.
While the orthogonal splitting in \cite{VY} is in general not explicit, our splitting is constructive. It is an interesting issue to study the information encoded by our decomposition, its relation with the variational representation of the transport coefficients \cite{KL99,Spohn} and in particular its relation with the finite dimensional approximation of the transport coefficients analyzed in \cite{AKM}. The decomposition in 2 dimensions introduces also some natural models for which the typical scaling of the current is expected to be composed by a gradient plus an orthogonal gradient of functions of the local density of particles.

The stationary condition for stochastic lattice gases can be naturally interpreted as an orthogonality condition with respect to a given harmonic discrete vector field. To satisfy this condition it is enough to consider discrete vector fields in the orthogonal complement and this can be done using the functional decomposition.

With both methods, we obtain very general classes and families of models naturally parameterized by functions on the configuration space and having similar structures. We are not going to compare the results with the two methods but in principle this is possible. In particular for any model the stationary condition can be interpreted with the two geometric constructions. We show only in a specific case how similar features in the rates appear.

For simplicity, we discuss all our results in dimensions 1 and 2 and for stochastic lattice gases satisfying an exclusion rule. All the arguments can be directly generalized to higher dimensions and to more general state space. In particular in dimensions higher than 2 the same constructions can be done. The basic ideas are the same but more notation and more general geometric constructions are necessary. We concentrate on finite range translational invariant  measures but different situations with long range interactions and/or non translational invariant measures are also interesting. We consider mainly cycles obtained by local perturbations of configurations of particles but the case of non local cycles is also interesting.

\smallskip

The structure of the paper is the following. In Section 2 we introduce notation and discuss some basic tools in graph theory, discrete geometry and Markov processes. In Section 3 we generate non reversible stochastic lattice gases starting from elementary cycles in the configuration space. In Section 4 we prove the functional Hodge decomposition for translational covariant discrete vector fields. In Section 5 we generate non reversible stochastic lattice gases using this decomposition.

\section{Preliminaries}
In this section we fix notation and develop some tools that will be used in the paper.

\subsection{Graphs}
Let $(V, \mathcal E)$ be a finite un-oriented graph without loops. This means that $|V|<+\infty$ and a generic element
of $\mathcal E$, called un-oriented edge or simply an edge, is $\{x,y\}$ a subset of cardinality $2$ of $V$. To every un-oriented graph $(V, \mathcal E)$ we associate canonically an oriented graph $(V,E)$ such that the set of oriented edges $E$ contains all the ordered pairs $(x,y)$ such that $\{x,y\}\in\mathcal E$. Note that if $(x,y)\in E$ then also $(y,x)\in E$. If $e=(x,y)\in E$ we denote
$e^-:=x$ and $e^+:=y$ and we call $\mathfrak e:=\{x,y\}$ the corresponding un-oriented edge. A sequence $(z_0,z_1,\dots ,z_k)$ of elements of $V$ such that $(z_i,z_{i+1})\in E$, $i=0,\dots k-1$, is called an oriented  path, or simply a path. A path with distinct vertices except $z_0=z_k$ is called a cycle. If $C=(z_0,z_1,\dots ,z_k)$ is a cycle and there exists an $i$ such that $(x,y)=(z_i,z_{i+1})$ we write $(x,y)\in C$. Likewise if there exists an $i$ such that $x=z_i$ we write $x\in C$. Two cycles that contain the same collection of oriented edges and differ by the starting point will be identified. We call $\mathcal C$ the collection of all the equivalence classes of finite length cycles and call $C$ a generic element.

A discrete vector field $\phi$ on $(V,\mathcal E)$ is a map $\phi:E\to \mathbb R$ such that $\phi(x,y)=-\phi(y,x)$.
A discrete vector field is of gradient type if there exists a function $g:V\to \mathbb R$ such that
$\phi(x,y)=[\nabla g](x,y):= g(y)-g(x)$.
The divergence of a discrete vector field $\phi$ at $x\in V$ is defined by
$$
\nabla\cdot \phi(x):=\sum_{y\,:\, \{x,y\}\in \mathcal E}\phi(x,y)\,.
$$
We call $\Lambda^1$ the $|\mathcal E|$ dimensional vector space of discrete vector fields. We endow $\Lambda^1$ with the scalar product
\begin{equation}\label{sc}
\langle \phi,\psi\rangle:=\frac 12\sum_{(x,y)\in E}\phi(x,y)\psi(x,y)\,, \qquad \phi,\psi\in \Lambda^1\,.
\end{equation}

A flow on an oriented graph is a map $Q:E\to \mathbb R^+$ that associates the amount of mass flowed $Q(x,y)$ to any edge $(x,y)\in E$. The divergence of a flow $Q$ at site $x\in V$ is defined as
$$
\nabla\cdot Q(x):=\sum_{y:(x,y)\in E}Q(x,y)-\sum_{y:(y,x)\in E}Q(y,x)\,.
$$
Given a cycle $C$ we introduce an elementary flow associated with the cycle by
\begin{equation}\label{lavas}
Q_C(x,y):=\left\{
\begin{array}{ll}
1   & \textrm{if}\ (x,y)\in C\,,  \\
0 & \textrm{if} \ (x,y)\not\in C\,.
\end{array}
\right.
\end{equation}
A similar definition can be given also for a path. Since for each $x\in C$ the outgoing flux is equal to ingoing flux and it is equal to $1$ we have $\nabla\cdot Q_C=0$.

\subsection{Discrete Hodge decomposition}\label{disch}

We recall briefly the Hodge decomposition for discrete vector fields. We consider the case when the graph $(V,\mathcal E)$ is the two dimensional discrete torus of size $N$. This means that the vertices are $\mathbb Z^2_N:=\mathbb Z^2/N\mathbb Z^2$ and the edges are the pairs of nearest neighbors sites. We call $e^{(i)}$, $i=1,2$ the vectors of the canonical basis, i.e. $e^{(1)}=(1,0)$ and $e^{(2)}=(0,1)$. We call $\Lambda^0$ the collection of real valued function defined on the set of vertices
$\Lambda^0:=\{g\, :\, \mathbb Z_N^2\to \mathbb R\}$. We recall also that $\Lambda^1$ denotes the vector space of discrete vector fields
endowed with the scalar product \eqref{sc}. Finally we call $\Lambda^2$ the vector space of 2-forms defined on the faces of the lattice $\mathbb Z^2_N$. Let us define this precisely. An oriented face is for example an elementary cycle in the graph of the type $(x,x+e^{(1)}, x+e^{(1)}+e^{(2)},x+e^{(2)},x)$ (recall that a cycle is an equivalence class modulo cyclic permutations). In this case we have an anticlockwise oriented face. This corresponds geometrically to a square having vertices $x,x+e^{(1)},x+e^{(1)}+e^{(2)},x+e^{(2)}$ plus an orientation in the anticlockwise sense. The same elementary face can be oriented clockwise and this corresponds to the elementary cycle $(x,x+e^{(2)}, x+e^{(1)}+e^{(2)},x+e^{(1)},x)$. If $f$ is a given oriented face we denote by $-f$ the oriented face corresponding to the same geometric square but having opposite orientation. A 2-form is a map $\psi$ from the set of oriented faces $F$ to $\mathbb R$ that is antisymmetric with respect to the change of orientation i.e. such that $\psi(-f)=-\psi(f)$. The boundary $\delta\psi$ of a 2 form $\psi$ is a discrete vector field defined by
\begin{equation}\label{defd}
\delta\psi(e):=\sum_{f\,:\, e\in f}\psi(f)\,.
\end{equation}
Since an oriented face is essentially an elementary oriented cycle the meaning of $e\in f$ has been already discussed.
By construction $\nabla\cdot \delta\psi=0$ for any $\psi$.
The 2 dimensional discrete Hodge decomposition \cite{B,GV,L} is written as the direct sum
\begin{equation}\label{dish2}
\Lambda^1=\nabla \Lambda^0\oplus \delta\Lambda^2\oplus\Lambda^1_H\,,
\end{equation}
where the orthogonality is with respect to the scalar product \eqref{sc}. The discrete vector fields on $\nabla\Lambda^0$ are the gradient ones. The dimension of $\nabla\Lambda^0$ is $N^2-1$. The vector subspace $\delta\Lambda^2$ contains all the discrete vector fields that can be obtained by \eqref{defd} from a given 2 form $\psi$. The dimension of $\delta\Lambda^2$ is $N^2-1$. Elements of $\delta\Lambda^2$ are called circulations. The dimension of $\Lambda_H^1$ is simply 2. Discrete vector fields in $\Lambda^1_H$ are called harmonic. A basis in $\Lambda^1_H$ is given by the vector fields $\varphi^{(1)}$ and $\varphi^{(2)}$ defined by
\begin{equation}\label{cenes}
\varphi^{(i)}\left(x,x+e^{(j)}\right):=\delta_{i,j}\,, \qquad i,j=1,2\,.
\end{equation}
The decomposition \eqref{dish2} holds in any dimension. For the $d$ dimensional torus the dimension of the
harmonic component is $d$. Starting from $d=3$ the combinatorial structure of the faces in the lattice is more complex
and a discussion of such cases would require much more space. Apart the more complex combinatorial structure there are however
no new phenomena involved and all our results can be formulated also in these cases.
Given a vector field $\phi\in \Lambda^1$ we write
\begin{equation}\label{hooo}
\phi=\phi^\nabla+\phi^\delta+\phi^H
\end{equation}
to denote the unique splitting in the three orthogonal components. This decomposition can be computed as follows. The harmonic part is determined writing $\phi^H=c_1\varphi^{(1)}+c_2\varphi^{(2)}$. The coefficients $c_i$ are determined by
\begin{equation}\label{cco}
c_i=\frac{1}{N^2}\sum_{x\in \mathbb Z^2_N}\phi\left(x,x+e^{(i)}\right)\,.
\end{equation}
To determine the gradient component $\phi^\nabla$ we need to determine a function $h$ for which $\phi^\nabla(x,y)=[\nabla h](x,y)=h(y)-h(x)$. This is done just taking the divergence on both side of \eqref{hooo} obtaining that $h$ solves the discrete Poisson equation $\nabla\cdot\nabla h=\nabla\cdot \phi$. The remaining component $\phi^\delta$ is computed just by difference $\phi^\delta=\phi-\phi^\nabla-\phi^H$.
We refer to \cite{B,GV,L} for respectively an algebraic, a short and a geometric perspective.

\smallskip
Given an oriented edge $e$ or an oriented face $f$ we denote respectively by $\mathfrak e$, $\mathfrak f$ the corresponding un-oriented edge and face. They will be often considered as sets, either of vertices or edges.
We use then naturally the symbols $\mathfrak e^c, \mathfrak f^c$ to denote the complementary sets. Note that for example both $f$ and $-f$ are associated with the same un-oriented face $\mathfrak f$.

\subsection{Markov chains and  particle systems}
We will consider finite and irreducible continuous time Markov chains with transitions rates $r$. If $(V,E)$ is the transition graph, then for $(x,y)\in E$ we have that $r(x,y)>0$ is the rate of jump from $x\in V$ to $y\in V$. There exists an unique invariant measure $\pi$ that is strictly positive and is characterized by the stationarity condition
\begin{equation}\label{stsu}
\sum_{y:(x,y)\in E}\pi(x)r(x,y)=\sum_{y:(y,x)\in E}\pi(y)r(y,x)\,.
\end{equation}
A Markov chain is reversible when it is satisfied detailed balance condition
$$\pi(x)r(x,y)=\pi(y)r(y,x)\,, \qquad \forall (x,y)\in E\,. $$
When the Markov chain is not reversible it is possible to define a time reversed Markov chain with transition rates defined by
$$
r^*(x,y):=\frac{\pi(y)r(y,x)}{\pi(x)}\,.
$$
This process is characterized by the following feature. In the stationary case it gives to a set of trajectories exactly the same probability that the original process gives to the set of time reversed trajectories. In particular the time reversed chain has the same invariant measure $\pi$ of the original process.

\smallskip

A special class of Markov chains that we will consider are stochastic particle systems whose configuration space is the collection of
configurations of indistinguishable particles located on the vertices of a graph $(V,E)$ \cite{Li}. We will denote by
 $\eta$ such a configuration and $\eta_t(x)$ will be the number of particles at site
$x\in V$ at time $t$.  The configuration space is $\Gamma^V$ where $\Gamma\subseteq \mathbb R$ is the set of possible configuration of particles in a single site. For simplicity we will consider the case $\Gamma=\{0,1\}$ that corresponds to the presence of  an exclusion rule. All the constructions can be straightforwardly generalized to more general state spaces. The particles jump randomly on the lattice and consequently the system jumps with rate $r(\eta,\eta')$ from the configuration $\eta$ to the configuration $\eta'$. The transition graph has the set of vertices coinciding with $\Gamma^V$ and edges corresponding to pairs of configurations $(\eta,\eta')$ such that $r(\eta,\eta')>0$. This rate is positive only when $\eta'$ is obtained from $\eta$ by a local modification of the configuration, typically one particle jumps from one site to a nearest neighbor one or one particle is created or annihilated in one single vertex.

Accordingly we introduce the  notation $\eta^x$, $\eta^{x,y}$ and $\eta^{\mathfrak e}$ where $x,y\in V$ and $\mathfrak e\in \mathcal E$. The notation $\eta^x$ is used when $\Gamma=\{0,1\}$ and denotes a configuration of particles obtained changing the value in the single site $x$. The configuration $\eta^{x,y}$ is obtained moving one particle from $x$ to $y$. The configuration $\eta^{\mathfrak e}$ is obtained exchanging the values at the two endpoints of $\mathfrak e\in \mathcal E$.  More precisely we have
\begin{equation}\label{afterj1}
\eta^{x}(z):=\left\{
\begin{array}{ll}
1-\eta(x) & \textrm{if}\ z=x \\
\eta(z) & \textrm{if}\ z\neq x\,,
\end{array}
\right.
\end{equation}
\begin{equation}\label{afterj2}
\eta^{x,y}(z):=\left\{
\begin{array}{ll}
\eta(x)-1 & \textrm{if}\ z=x \\
\eta(y)+1 & \textrm{if}\ z=y\\
\eta(z) & \textrm{if}\ z\neq x,y
\end{array}
\right.
\end{equation}
\begin{equation}\label{afterj3}
\eta^{\mathfrak e}(z):=\left\{
\begin{array}{ll}
\eta(e^-) & \textrm{if}\ z=e^+ \\
\eta(e^+) & \textrm{if}\ z=e^-\\
\eta(z) & \textrm{if}\ z\neq e^\pm\,.
\end{array}
\right.
\end{equation}
In \eqref{afterj3} $e$ is any orientation of $\mathfrak e$.
We call respectively $c_{x}(\eta)$, $c_{x,y}(\eta)$ and $c_{\mathfrak e}(\eta)$ the corresponding rate of transitions. For example we have $c_x(\eta):=r(\eta,\eta^x)$ and similarly for the other cases.

On the discrete torus acts naturally the group of translations. We denote by $\tau_x$ the translation by the element $x\in \mathbb Z^2_N$. The translations act on configurations by $\left[\tau_x\eta\right](z):=\eta(z-x)$ and on functions by $\left[\tau_x g\right](\eta):=g(\tau_{-x}\eta)$. For notational convenience it is useful to define $\tau_{\mathfrak f}$ for an un-oriented face $\mathfrak f$. If the vertices belonging to $\mathfrak f$ are $\{x,x+e^{(1)},x+e^{(2)},x+e^{(1)}+e^{(2)}\}$ then we define $\tau_{\mathfrak f}:=\tau_x$.
Likewise for an un-oriented edge $\mathfrak e=\{x,x+e^{(i)}\}$, $i=1,2$  we define $\tau_{\mathfrak e}:=\tau_x$.

We call $\mathbb Z_N^{*2}$ the dual lattice of $\mathbb Z_N^2$. The vertices of $\mathbb Z_N^{*2}$ are in the centers of the faces of the original lattice $\mathbb Z_N^2$. This means that every face of the original lattice is naturally  associated with a vertex $x^*\in \mathbb Z_N^{*2}$. There is a natural correspondence also between the un-oriented edges $\mathcal E$ and $\mathcal E^*$ of the two lattices since pairs of dual edges cross each other in the center.

Given a configuration $\eta$ and a subset $W\subseteq \mathbb Z^2_N$ we denote by $\eta_W$
the restriction of the configuration to the subset $W$. Given two configurations $\eta,\xi$ and two
subsets $W,W'$ such that $W\cap W'=\emptyset$, we denote by $\eta_W\xi_{W'}$ the configuration of particles on $W\cup W'$ that coincides with $\eta$ on $W$ and coincides with $\xi$ on $W'$. A function $h:\Gamma^{\mathbb Z^2_N}\to \mathbb R$ is called local if there exists a finite subset $W$ (independent of $N$ large enough) such that $h\left(\eta_W\xi_{W^c}\right)=h(\eta)$ for any $\eta,\xi$, where $W^c$ is the complementary set of $W$. The minimal subset $W$ for which this holds is called the domain of dependence of the function $h$ and is denoted by $D(h)$.
When the dependence of the transition rates on the configuration is local we say that the dynamics is \emph{local}.
We will consider only translational covariant models for which the transition mechanism is the same in every point of the lattice.
This is formalized by requiring that there exists local functions $c_0(\eta)$, $c_{0,\pm e^{(i)}}(\eta)$ and $c_{\mathfrak{e}^{(i)}}(\eta)$ such that $c_x(\eta)=\tau_xc_0(\eta)$, $c_{x,x\pm e^{(i)}}(\eta)=\tau_xc_{0,\pm e^{(i)}}(\eta)$ and $c_{\mathfrak e}(\eta)=\tau_{\mathfrak e}c_{\mathfrak{e}^{(i)}}(\eta)$ for an edge of the type $\mathfrak e=\{x,x\pm e^{(i)}\}$.

All the information on the stochastic evolution are encoded in the generator of the dynamics
that act on a function $f$ as
\begin{equation}\label{generaleL}
\mathcal L f(\eta)=\sum_{\eta'}r(\eta,\eta')\left[f(\eta')-f(\eta)\right]\,.
\end{equation}
Given a configuration of particles $\eta\in \{0,1\}^V$ we call $\mathfrak C(\eta)$ the collection of clusters of particles
that is induced on $V$.
A cluster $c\in \mathfrak C(\eta)$ is a subgraph of $(V,\mathcal E)$. Two sites $x,y\in V$ belong to the same cluster $c$ if
$\eta(x)=\eta(y)=1$ and there exists an un-oriented path $(z_0,z_1, \dots ,z_k)$ such that $\eta(z_i)=1$ and
$(z_i,z_{i+1})\in \mathcal E$.

\section{Invariant measures and divergence free flows}\label{imdff}

The problem of determining the invariant measure of a given Markov chain and the problem of construct a divergence free flow on its transition graph are strictly related.
Suppose that we have an irreducible Markov chain with transition rates $r$ and invariant measure $\pi$. Then on the transition graph we can define the flow
\begin{equation}\label{defbas}
Q(x,y):=\pi(x)r(x,y)\,.
\end{equation}
The stationarity condition \eqref{stsu} coincides with the divergence free condition for the flow $Q$. Conversely suppose that we have a divergence free flow $Q$ and a fixed target probability measure $\pi$ strictly positive. We obtain a Markov chain having invariant measure $\pi$ defining the rates as
\begin{equation}\label{qpi}
r(x,y):=\frac{Q(x,y)}{\pi(x)}\,.
\end{equation}
Indeed all the Markov chains having $\pi$ invariant are obtained in this way for a suitable divergence free flow.
Once again the proof follows inserting the rates \eqref{qpi} into \eqref{stsu} and using the divergence free condition of $Q$. The natural interpretation of the flow $Q$ is that it represents the typical flow observed in the stationary state of the chain.

In terms of flows, the connection between a Markov chain and its time reversed is even more transparent. Consider a Markov chain having transition rates $r$ and invariant measure $\pi$ and let $Q$ be defined by \eqref{defbas}. The time reversed flow $Q^*$ is defined simply reversing the original flow $Q$, i.e. we have $Q^*(x,y):=Q(y,x)$. It is clear that $Q^*$ is still a divergence free flow. The rates of the time reversed Markov chain are obtained by \eqref{qpi} but using the reversed flow and the same invariant measure $\pi$: $r^*(x.y)=\frac{Q^*(x,y)}{\pi(x)}$. In particular reversibility corresponds to the symmetry by inversion $Q(x,y)=Q(y,x)$.

This connection between the two problems is important since the geometric structure of divergence free flows is quite well understood and there are simple representations theorems \cite{BFG,GV,MC}. We have indeed that on a finite oriented graph any divergence free flow can be written as a superposition of elementary flows associated with cycles
\begin{equation}\label{Qr}
Q=\sum_{C\in \mathcal C}\rho(C)Q_C\,,
\end{equation}
for suitable positive weights $\rho$. This decomposition is in general not unique and under suitable assumptions on the flow (for example summability) the result can be extended also to infinite graphs \cite{BFG}. For a symmetric flow for which $Q(x,y)=Q(y,x)$ the decomposition \eqref{Qr} can be done using just cycles of length 2 of the form $C=(x,y,x)$.

Using this simple construction it is possible to generate non trivial and interesting non reversible Markov dynamics.
We start with the simplest case that is the reversible one. After we discuss the simplest non reversible
model that is the case of one single particle performing a random walk on a lattice. We then start to discuss some particle models.
Since any Markovian model can be generated by a suitable choice of the cycles we concentrate on a specific class of models. In particular we consider lattice gases having a Gibbsian invariant measure and such that the associated divergence free flow
can be constructed using cycles obtained by local perturbations of particles configurations. The number of cycles on the configuration space is huge and their combinatorial structure is very rich. In a sense we are considering the simplest possible class of models.
Clearly, even restricting to local cycles, it is possible to consider many other models.
At the end we briefly discuss models with local rates and corresponding to global cycles where particles go around all the system. For simplicity we consider stochastic lattice gases but more general state space can be discussed similarly. In this case the combinatorial structure of the cycles is even richer.

We discuss our examples of particles systems following an order that may appear strange. Our motivation is the following. We start discussing the two dimensional Kawasaki dynamics since in this case the link between the combinatorial structure of the cycles in the configuration space and the structure of the physical space is more evident. Every cycle is indeed naturally associated to a two dimensional face of the physical lattice. We then discuss other examples of stochastic lattice gases that should enlightening the general idea. After this we describe
the general structure of our approach based on local perturbations of particles configurations. We end describing some natural models
having a decomposition in cycles of different type.

A very natural and important issue is to understand the relationship between the structure of the cycles
and the behavior of the system, for example the speed of convergence towards the invariant measure.

In the continuous case again the problem of finding invariant measures and the problem
of constructing divergence free currents are strictly related. In this case all the machinery
of the cycles is much more subtle and less natural. However there are several different ways of
generating divergence free currents and not reversible dynamics could be constructed similarly to the discrete case.

\subsection{Metropolis algorithm} The first example is the classic Metropolis algorithm \cite{Metr} that corresponds to the generation
of reversible dynamics.
As we observed the reversible chains have a symmetric typical flow Q. This means that fixing the rates as in \eqref{qpi}
with a symmetric flow $Q$ we obtain a reversible chain with invariant measure $\pi$. For example choosing
$Q(x,y)=\max\{\pi(x),\pi(y)\}$, that is clearly symmetric, gives the classic rates
$$
r(x,y)=\max\left\{1,\frac{\pi(y)}{\pi(x)}\right\}\,.
$$

\subsection{Random walks} Before consider the case of several particles let us shortly recall a construction in \cite{GV} for one single particle. We consider the bidimensional case although the arguments presented below can be generalized to any dimension. Consider
for example  two  nonnegative (i.e. with nonnegative coordinates) vector fields $q^+(x)$ and $q^-(x)$ on the continuous torus $[0,N]\times [0,N]$ with periodic boundary conditions and such that
\begin{equation}
j(x):=q^+(x)-q^-(x)\,,
\label{diffj}
\end{equation}
is a  divergence free vector field.

We define a flow $Q$ on the lattice fixing the values  $Q\left(y,y  \pm e^{(i)}\right)$  as the flux of $q^\pm_i$ across  the edges dual to $\left\{y,y\pm e^{(i)}\right\}$, more precisely
\begin{align}
Q\left(y,y+e^{(1)}\right):= &\int_{0}^{1} q^+_1\left(y+ \frac{e^{(1)}-e^{(2)}}{2} + \alpha e^{(2)}
\right)d\alpha\,,\\
Q(y,y+e^{(2)}):= &\int_{0}^{1}   q^+_2\left(y+ \frac{-e^{(1)}+e^{(2)}}{2} + \alpha e^{(1)}  \right)d\alpha\,,\\
Q(y,y-e^{(1)}):=& \int_{0}^{1}    q^-_1\left(y- \frac{e^{(1)}+e^{(2)}}{2} + \alpha e^{(2)}
\right)d\alpha\,,\\
Q(y,y-e^{(2)}):=& \int_{0}^{1} q^-_2\left(y  - \frac{e^{(1)}+e^{(2)}}{2} + \alpha e^{(1)}
\right)d\alpha\,.
\end{align}
Given $y\in \mathbb Z_N^2$ we have that $\nabla\cdot Q (y)$ coincides with the flow (from inside to outside) of $j$ through the boundary of the box $B_y:=\{z \in \bb Z^2 \,:\, |y-z|_\infty \leq 1/2\}$
$$
\nabla\cdot Q(y)=\int_{\partial B_y}j\cdot \hat n\, d\Sigma=0\,.
$$
The last equality above follows by $\nabla\cdot j=0$ and the  Gauss-Green Theorem.

If we define the rate of transitions of a random walk by \eqref{qpi} we obtain a random walk with invariant measure $\pi$.

\subsection{Two dimensional Kawasaki dynamics }\label{basile}
We consider a system of particles satisfying an exclusion rule, i.e. the configuration on one single site is $\Gamma=\{0,1\}$, and evolving on $\mathbb Z^2_N$. The construction can be naturally generalized to any planar graph. Generalizations to higher dimensions are also natural but require more notation.

Let $\mu$ be a Gibbsian probability measure on the configuration space $\Gamma^{\mathbb Z^2_N}$
given by
\begin{equation}\label{gibbs}
\pi(\eta)=\frac 1Z e^{-H(\eta)}\,,
\end{equation}
where $H$ is the Hamiltonian.
To explain notation we consider the case of an Hamiltonian with one body  and two body interactions
\begin{equation}\label{H}
H(\eta)=- \sum_{\{x,y\}\in \mathcal E}J_{\{x,y\}}\eta(x)\eta(y)-\sum_{x\in \mathbb Z^2_N}\lambda_x\eta(x)\,.
\end{equation}
We consider such an Hamiltonian for simplicity but more general interactions can be handled similarly. Indeed the specific form of the Hamiltonian does not play any role in the following. It is only important that it has interactions of bounded range.
In \eqref{H} the parameters $\left(J_{\{x,y\}}\right)_{\{x,y\}\in \mathcal E}$ and $\left(\lambda_x\right)_{x\in \mathbb Z^2_N}$ are arbitrary real numbers describing respectively the interactions associate to the bonds and the chemical potentials of the sites. In the following we will always restrict to translational invariant Hamiltonian.
Given $W \subseteq \mathbb Z^2_N$ we define the energy
restricted to $W$ as
\begin{equation}\label{en-ris}
H_{W}(\eta):=-\sum_{\{x,y\}\cap W \neq \emptyset}J_{\{x,y\}}\eta(x)\eta(y)-\sum_{x\in W}\lambda_x\eta(x)\,.
\end{equation}
We define also
\begin{equation}\label{en-ris*}
H_{W}^*(\eta):=-\sum_{\{x,y\}\subset W }J_{\{x,y\}}\eta(x)\eta(y)-\sum_{x\in W}\lambda_x\eta(x)\,.
\end{equation}
For any $W$ we have $H=H_W+H^*_{W^c}$.
Given an oriented edge $e\in  E$ of the lattice there is only one anticlockwise oriented face to which
$e$ belongs that we call $f^+(e)$. There is also an unique anticlockwise face, that we call $f^-(e)$, such that $e\in -f^-(e)$
(see Figure 1).
\begin{figure}[]
\centering
\hspace{0.1cm}
\entrymodifiers={+<0.5ex>[o][F*:black]}
\xymatrix@C=1.5cm@R=1.5cm{
 {}\ar@{.}[rrrr]\ar@{.}[dddd] & \ar@{.}[dddd] & {}\ar@{.}[dddd] & {}\ar@{.}[dddd] & {}\ar@{.}[dddd]\\
 {}\ar@{.}[rrrr] & {} & {} & {} & {}\\
 {}\ar@{.}[rrrr]_(.53){\textit{\normalsize y}} & {} & {} & {} & {}\\
 {}\ar@{.}[rrrr]_(.53){\textit{\normalsize x}}^(0.38){\textit{\tiny{$f^{+}(e)$}}} ^(0.63){\textit{\tiny{$f^{-}(e)$}}} & {} \ar@{}[ur]|{\textit{\Huge{$\circlearrowleft$}}} & {} \ar@{->}[u]|{\textit{}} & {} \ar@{}[ul]|{\textit{\Huge{$\circlearrowleft$}}} & {}\\
  {} \ar@{->}[u]^(.5){\textit{\normalsize $e^{(2)}$}}  \ar@{->}[r]_(.5){\textit{\normalsize $e^{(1)}$}} \ar@{.}_(-0.05){\text{\normalsize }}[rrrr] & {} & {} & {} & {}
}
\caption{\small{ The discrete two dimensional torus of side 5; the opposite sides of the square are identified. For the oriented edge $(x,y)=e$ we draw the two anticlockwise oriented faces $f^{-}(e)$ and $f^{+}(e)$. }}
\label{fig: due facce}
\end{figure}
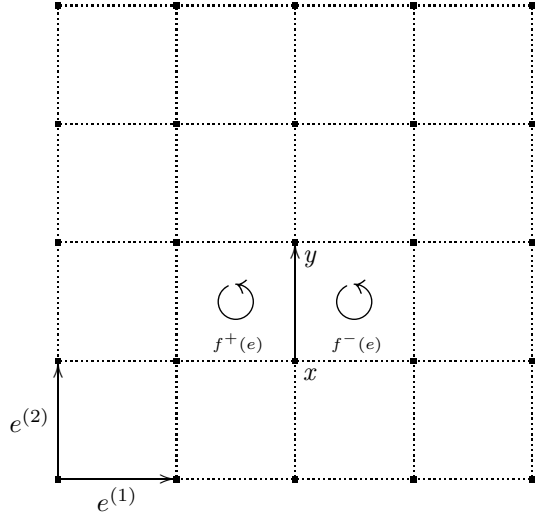

We will show that if we define the rates of jump like in \eqref{iratesgen}, where the two non-negative functions $w^\pm$ are such that $D(w^\pm)\cap \{0,e^{(1)},e^{(2)}, e^{(1)}+e^{(2)}\}=\emptyset$, then we will obtain a dynamics having \eqref{gibbs}
as invariant measure.

Let us introduce before a simplified version of models.
Let $w^+,w^-$ two positive numbers.
We consider the following transition rates for the jump of one particle from $e^-$ to $e^+$
\begin{equation}\label{irates}
c_{e^-,e^+}(\eta):=\eta(e^-)(1-\eta(e^+))\left(w^+e^{H_{\mathfrak f^+(e)}(\eta)}+w^-e^{H_{\mathfrak f^-(e)}(\eta)}\right)\,.
\end{equation}
When we write $H_{\mathfrak f}$ for a suitable face $\mathfrak f$ we consider the face as a set of vertices.
We claim that the generator \eqref{generaleL} with the rates of jump \eqref{irates}  has \eqref{gibbs}
as invariant measure, moreover if $w^-\neq w^+$ then the dynamics is not reversible. Since the Hamiltonian \eqref{H} has only finite range interactions the dynamics induced by \eqref{irates} is local.

First we give a direct proof of this claim and then we show how a generalized version of the rates \eqref{irates} is naturally constructed by divergence free flows on the configuration space.
We need to check the validity of the stationary equations
\begin{equation}\label{staz}
\sum_{e\in E}\Big[\pi(\eta)c_{e^-,e^+}(\eta)-\eta(e^-)(1-\eta(e^+))\pi\left(\eta^{e^-,e^+}\right)c_{e^+,e^-}\left(\eta^{e^-,e^+}\right)\Big]=0\,,\qquad \forall \eta\,.
\end{equation}
Using \eqref{gibbs} and \eqref{irates} we will obtain
in \eqref{staz} a sum of several terms having as a factor
$e^{-H^*_{\mathfrak f^c}(\xi)}$ for some face $\mathfrak f$ and some configuration $\xi$.
Observe that if $\eta$ and $\xi$ are obtained one from the other just changing
the occupation numbers on sites belonging to  $\mathfrak f$ then we have $H^*_{\mathfrak f^c}(\eta)=H^*_{\mathfrak f^c}(\xi).$ This means that all the factors can be written as $e^{-H^*_{\mathfrak f^c}(\eta)}$ for different $\mathfrak f$.
We will use the relationship $f^+(e)=f^-(-e)$ (with $-e$ we denote the edge oriented oppositely with respect to $e$).

Let $\mathcal F$ be the collection of  un-oriented faces. We group together all the terms in \eqref{staz}
that have the energetic factor equal to $e^{-H^*_{\mathfrak f^c}}$  for a given $f\in \mathcal F$.
The sum of all these terms is equal to
\begin{equation}\label{arpa}
\frac{e^{-H^*_{\mathfrak f^c}(\eta)}}{Z}\left\{\sum_{e \in f}\Big[\eta(e^-)(1-\eta(e^+))(w^+-w^-)+\eta(e^+)(1-\eta(e^-))(w^--w^+)\Big]\right\}\,.
\end{equation}
In \eqref{arpa} $f$ is the unique anticlockwise oriented face corresponding to $\mathfrak f$.
The sum appearing in
\eqref{arpa} is zero  since coincides with the telescopic sum
\begin{equation}\label{sz}
(w^+-w^-)\sum_{e\in f}\left[\eta(e^-)-\eta(e^+)\right]=0\,.
\end{equation}
Considering all the faces in $\mathcal F$ we can write \eqref{staz} as
\begin{equation}\label{fdz}
\sum_{\mathfrak f\in \mathcal F}\left\{\frac{e^{-H_{\mathfrak f^c}^*(\eta)}}{Z}(w^+-w^-)\sum_{e\in f}\left[\eta(e^-)-\eta(e^+)\right]\right\}=0\,.
\end{equation}
By \eqref{sz} we have that \eqref{fdz} is  satisfied.
The reversibility condition becomes
\begin{equation}
(w^+-w^-)e^{-H_{\mathfrak f^{-}(e)^c}^*(\eta)}=(w^+-w^-)e^{-H_{\mathfrak f^{+}(e)^c}^*(\eta)},
\end{equation}
that is not satisfied when $w^+\neq w^-$ apart special cases.

\smallskip

We show now how it is possible to conceive a generalized version of \eqref{irates} constructing a divergence free flow like in \eqref{Qr} on the configuration space
and using then \eqref{qpi}. We use cycles for which the configuration of particles is frozen outside a given face $\mathfrak f$.  The cycles that we use are shown in Figure \ref{fig: EP 2d} and correspond to letting the particles rotate around a fixed face according to specific rules and following the two orientations. Let us consider two non-negative functions $w^\pm$ such that
$D(w^\pm)\cap \{0,e^{(1)},e^{(2)}, e^{(1)}+e^{(2)}\}=\emptyset$. If $w^\pm$ are local then the corresponding rates will be also local.

\begin{remark}\label{zeb}
More generally we can assume that $w^\pm$ depends on the configuration of particles restricted to  $\{0,e^{(1)},e^{(2)}, e^{(1)}+e^{(2)}\}$ only through the number of particles
$\eta(0)+\eta(e^{(1)})+\eta(e^{(2)})+\eta(e^{(1)}+e^{(2)})$.
\end{remark}
\begin{figure}[]

\[
\hspace{0.1cm}
\entrymodifiers={+<0.5ex>[o][F*:black]}
\xymatrix@C=1.2cm@R=1.2cm{
  \ar@{.}[r] \ar@{}[dr]|{f}\ar@{}[rrrd]|{\textit{\large{$\rightarrow$}}} &   \ar@{.}[d] &   \ar@{.}[d] \ar@{.}[r]\ar@{}[rrrd]|{\textit{\large{$\rightarrow$}}} \ar@{}[dr]|{f} &  \ar@{.}[l]  &  \ar@{.}[r] \ar@{}[rd]|{f}\ar@{}[rrrd]|{\textit{\large{$\rightarrow$}}} &  *+[o][F*:black]{}\ar@{.}[d] &  *+[o][F*:black]{}\ar@{.}[r] \ar@{.}[d] &  \ar@{.}[d] \ar@{}[ld]|{f}  \\
  *+[o][F*:black]{} \ar@{.}[u]  &  \ar@{.}[l]   &   \ar@{.}[r]  & *+[o][F*:black]{}  \ar@{.}[u] \ar@{.} [u] &  \ar@{.}[u] \ar@{.}[r]  &  &  \ar@{.}[r] &   \ar@{.}[u]\\
  *+[o][F*:black]{}\ar@{.}[r] \ar@{}[dr]|{f} \ar@{}[rrrd]|{\textit{\large{$\rightarrow$}}}&   \ar@{.}[d] &   *+[o][F*:black]{}\ar@{.}[d] \ar@{.}[r] \ar@{}[dr]|{f}\ar@{}[rrrd]|{\textit{\large{$\rightarrow$}}} &  *+[o][F*:black]{}\ar@{.}[l]  &  \ar@{.}[r] \ar@{}[rd]|{f}\ar@{}[rrrd]|{\textit{\large{$\rightarrow$}}} &  *+[o][F*:black]{}\ar@{.}[d] &  \ar@{.}[r] \ar@{.}[d] & *+[o][F*:black]{} \ar@{.}[d] \ar@{}[ld]|{f}  \\
   \ar@{.}[u]  &  *+[o][F*:black]{}\ar@{.}[l] &   \ar@{.}[r]  &  \ar@{.}[u] \ar@{.} [u] &  *+[o][F*:black]{}\ar@{.}[u] \ar@{.}[r]  &  &  \ar@{.}[r] &  *+[o][F*:black]{} \ar@{.}[u]\\
  \ar@{.}[r] \ar@{}[dr]|{f}\ar@{}[rrrd]|{\textit{\large{$\rightarrow$}}} &  *+[o][F*:black]{} \ar@{.}[d] &   *+[o][F*:black]{}\ar@{.}[d] \ar@{.}[r] \ar@{}[dr]|{f}\ar@{}[rrrd]|{\textit{\large{$\rightarrow$}}} &  \ar@{.}[l]  &  *+[o][F*:black]{}\ar@{.}[r] \ar@{}[rd]|{f}\ar@{}[rrrd]|{\textit{\large{$\rightarrow$}}} &  \ar@{.}[d] &  \ar@{.}[r] \ar@{.}[d] &  \ar@{.}[d] \ar@{}[ld]|{f}  \\
   *+[o][F*:black]{}\ar@{.}[u]  &  \ar@{.}[l] & *+[o][F*:black]{}  \ar@{.}[r]  &   \ar@{.}[u] \ar@{.} [u] &  \ar@{.}[u] \ar@{.}[r]  & *+[o][F*:black]{} &  *+[o][F*:black]{}\ar@{.}[r] &   *+[o][F*:black]{}\ar@{.}[u]\\
  *+[o][F*:black]{}\ar@{.}[r] \ar@{}[dr]|{f} \ar@{}[rrrd]|{\textit{\large{$\rightarrow$}}}&   \ar@{.}[d] &  *+[o][F*:black]{} \ar@{.}[d] \ar@{.}[r] \ar@{}[dr]|{f}\ar@{}[rrrd]|{\textit{\large{$\rightarrow$}}} & *+[o][F*:black]{} \ar@{.}[l]  &  *+[o][F*:black]{}\ar@{.}[r] \ar@{}[rd]|{f}\ar@{}[rrrd]|{\textit{\large{$\rightarrow$}}} & *+[o][F*:black]{}  *+[o][F*:black]{}\ar@{.}[d] &  \ar@{.}[r] \ar@{.}[d] & *+[o][F*:black]{} \ar@{.}[d] \ar@{}[ld]|{f}  \\
  *+[o][F*:black]{} \ar@{.}[u]  &  *+[o][F*:black]{}\ar@{.}[l] &  *+[o][F*:black]{} \ar@{.}[r]  &   \ar@{.}[u] \ar@{.} [u] &  \ar@{.}[u] \ar@{.}[r]  & *+[o][F*:black]{} &  *+[o][F*:black]{}\ar@{.}[r] &   *+[o][F*:black]{}\ar@{.}[u]\\
}
\]

\caption{\small{Anticlockwise cycles in $\mathcal C_f$. Each line corresponds to a cycle and the evolution is from left to right. The configuration of particles is frozen outside the face. In the first cycle there is 1 particle rotating in the face, in the second and third 2 and in the last one 3. The clockwise cycles are obtained reading from right to left the Figure. }}
\label{fig: EP 2d}
\end{figure}
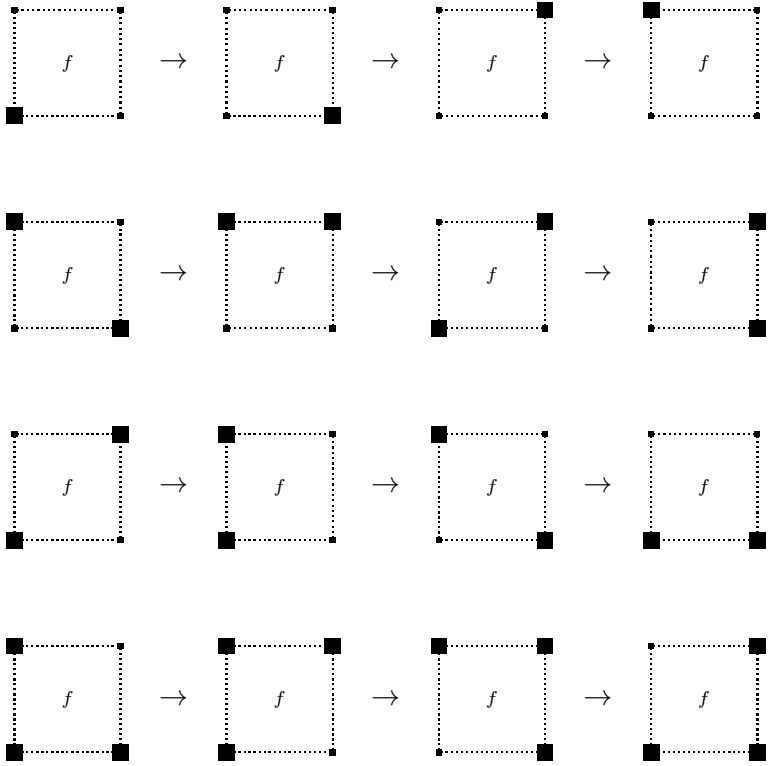
We call $\mathcal C_f$ the collection of cycles in the configuration space associated with the oriented face $f$. In Figure \ref{fig: EP 2d} we draw the structure of these cycles associated with an $f\in F^+$ (we call $F^+$ the collection of anticlockwise oriented faces). For different configurations of the particles outside $\mathfrak f$ we obtain different elements of $\mathcal C_f$.
The weights in \eqref{Qr} are fixed by
\begin{equation}\label{weights cycles kawa 2-d}
\rho(C):=\left\{
\begin{array}{ll}
e^{-H^*_{\mathfrak f^c}}\tau_{\mathfrak f}w^+ & \textrm{if}\  C\in \mathcal C_f\,, f\in F^+\,,\\
e^{-H^*_{\mathfrak f^c}}\tau_{\mathfrak f}w^- & \textrm{if}\  C\in \mathcal C_f\,, f\in F^-\,.
\end{array}
\right.
\end{equation}
A cycle $C\in \mathcal C_f$ is individuated by the face $f$, the type of rotation around the face among the possible ones
in Figure \ref{fig: EP 2d} and the configuration $\eta_{\mathfrak f^c}$ outside the face.
Definition \eqref{weights cycles kawa 2-d} is then well posed since all the functions appearing depend only on $\eta_{\mathfrak f^c}$
and are therefore constant on the cycle $C$.

The value of $Q\left(\eta,\eta^{x,y}\right)$ associated with a jump of one particle from $x$ to $y$ ($e=(x,y)$ is an oriented
edge of the lattice) in the configuration $\eta$ is determined as follows. This value is zero unless the site $x$ is occupied and the site $y$ is empty. This gives a factor $\eta(x)(1-\eta(y))$. If this constraint is satisfied then there is a contribution corresponding to $e^{-H^*_{\mathfrak f^+(e)^c}(\eta)}\tau_{\mathfrak f^+(e)}w^+(\eta)$ from a cycle with anticlockwise rotations and a contribution corresponding to
$e^{-H^*_{\mathfrak f^-(e)^c}(\eta)}\tau_{\mathfrak f^-(e)}w^-(\eta)$ from a cycle with clockwise rotations. Applying formula \eqref{qpi}
we obtain the following generalized version of the rates \eqref{irates}
\begin{equation}\label{iratesgen}
c_{e^-,e^+}(\eta)=\eta(e^-)(1-\eta(e^+))
\left(e^{H_{\mathfrak f^+(e)}(\eta)}\tau_{\mathfrak f^+(e)}w^+(\eta)+e^{H_{\mathfrak f^-(e)}(\eta)}\tau_{\mathfrak f^-(e)}w^-(\eta)\right)\,.
\end{equation}
The exponential factors in \eqref{weights cycles kawa 2-d} have been chosen in such a way that applying \eqref{qpi}
the non locality of the measure $\pi$ is erased and we obtain local rates.
To these rates it is always possible to add some reversible rates coming from cycles of length 2 of the
form $C=(\eta,\eta^{x,y},\eta)$. Since this happens also in one dimension we discuss this issue in the next section.

If $\pi$ is a Bernoulli measure then $e^{H_{\mathfrak f^\pm(e)}}$ depends only on the number of particles in the face $\mathfrak f^\pm(e)$. This dependence can be compensated by $w^\pm$ by Remark \ref{zeb}.

\subsection{One dimensional Kawasaki dynamics}\label{odkd}
We consider now a conservative dynamics on a ring with $N$ sites.
We search for local rates having (\ref{gibbs}) as  invariant  measure. Defining the set
$W(\{x,x+1\}):=\{x-1,x,x+1,x+2\}$, we will obtain the following rates
\begin{align}\label{euna}
&c_{x,x+1}(\eta)=\eta(x)(1-\eta(x+1))\Big\{e^{H_{W(\{x-1,x\})}}\Big[\eta(x-1)(1-\eta(x-2))\tau_{x-1}w^+ \nonumber\\
&+\eta(x-2)(1-\eta(x-1))\tau_{x-1}w^-\Big]+e^{H_{W(\{x+1,x+2\})}}\Big[\eta(x+2)(1-\eta(x+3))\tau_{x+1}w^+\\
&+\eta(x+3)(1-\eta(x+2))\tau_{x+1}w^-\Big]+e^{H_{\{x,x+1\}}}\tau_x w\Big\}\,,\nonumber
\end{align}
and
\begin{align}\label{edue}
& c_{x+1,x}(\eta)=\eta(x+1)(1-\eta(x))\Big\{e^{H_{W(\{x-1,x\})}}\Big[\eta(x-2)(1-\eta(x-1))\tau_{x-1}w^+ \nonumber\\
& +\eta(x-1)(1-\eta(x-2))\tau_{x-1}w^-\Big]+e^{H_{W(\{x+1,x+2\})}}\Big[\eta(x+3)(1-\eta(x+2))\tau_{x+1}w^+\\
&+\eta(x+2)(1-\eta(x+3))\tau_{x+1}w^-\Big]+e^{H_{\{x,x+1\}}}\tau_x w\Big\}\,,\nonumber
\end{align}
where $w, w^\pm$ are arbitrary non-negative functions such that $D(w)\cap \{0,1\}=\emptyset$ and $D(w^\pm)\cap W(\{0,1\})=\emptyset$.

The class of cycles that we consider contains all the cycles in which one single particle jumps across an edge and then come back. The length of these cycles is 2 and any superposition of them is generating a symmetric flow. The family of such cycles, when one single particle jumps across the edge $\mathfrak e$, is denoted by $\mathcal C_{\mathfrak e}^r$.

The most simple way to  introduce irreversibility  is to consider  cycles associated with two particles evolving.
The structure of the cycles associated with the movement of two particles that we consider is
illustrated in Figure \ref{fig: due v- due p Kaw1d}. Again the configuration outside the window drawn is frozen. Moving from the top to the bottom of the Figure we have a cycle in the collection $\mathcal C^+_{\mathfrak e}$. Moving instead from the bottom to the top we have a cycle in the collection $\mathcal C_{\mathfrak e}^-$. The lower index $\mathfrak e$ denotes the edge around which the evolution takes place while instead the upper index denotes the direction of movement. Observe that on cycles $\mathcal C_{\mathfrak e}^\pm$ particles are not jumping across the edge $\mathfrak e$ but
instead across $\tau_{\pm 1}\mathfrak e$.
\begin{figure}[]
\[
\entrymodifiers={+<0.5ex>[o][F*:black]}
\xymatrix@C=1.5cm@R=1.5cm{
*+[o][F*:black]{}
\ar@{.}[rrr]^{\textit{\normalsize $\mathfrak e$}} &\ar@{}[dr]|{\textit{\large{$\downarrow$}}}&*+[o][F*:black]{}&\\
\ar@{.}[rrr]^{\textit{\normalsize $\mathfrak e$}}&*+[o][F*:black]{}\ar@{}[dr]|{\textit{\large{$\downarrow$}}}&*+[o][F*:black]{}&\\
\ar@{.}[rrr]^{\textit{\normalsize $\mathfrak e$}}&*+[o][F*:black]{}\ar@{}[dr]|{\textit{\large{$\downarrow$}}}&&*+[o][F*:black]{}\\
*+[o][F*:black]{}\ar@{.}[rrr]^{\textit{\normalsize $\mathfrak e$}}&&&*+[o][F*:black]{}\\
} \]
\caption{\small{Following the arrows from top to bottom we obtain a cycle in the collection $\mathcal C_{\mathfrak e}^+$. Recall that the configuration is frozen outside the portion of the lattice drawn. Going in the opposite direction from bottom to top
we obtain a cycle in $\mathcal C_{\mathfrak e}^-$.}}
\label{fig: due v- due p Kaw1d}
\end{figure}
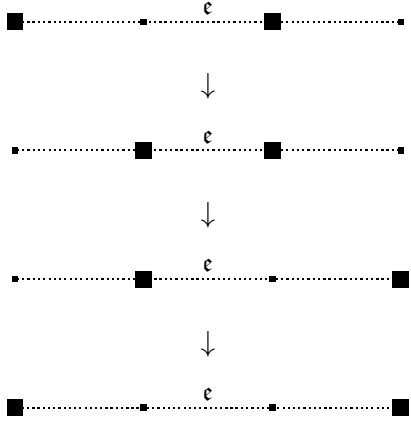
We call $W(\mathfrak e):=\tau_{-1}\mathfrak e\cup \mathfrak e\cup \tau_1\mathfrak  e$, i.e. the set of vertices belonging to one of the edges considered. We consider non negative functions $w, w^\pm$ such that $D(w)\cap \{0,1\}=\emptyset$ and $D(w^\pm)\cap W(\{0,1\})=\emptyset$. We construct a divergence free flow with a decomposition \eqref{Qr} with weights defined by
\begin{equation}\label{weights cycles kawa 1-d}
\rho(C):=\left\{
\begin{array}{ll}
e^{-H^*_{\mathfrak e^c}}\tau_{\mathfrak e} w & \textrm{if}\  C\in \mathcal C_{\mathfrak e}^r \,, \\
e^{-H^*_{W^c(\mathfrak e)}}\tau_{\mathfrak e} w^+ & \textrm{if}\ C\in \mathcal C_{\mathfrak e}^+\,,\\
e^{-H^*_{W^c(\mathfrak e)}}\tau_{\mathfrak e} w^- & \textrm{if}\ C\in \mathcal C_{\mathfrak e}^-\,.
\end{array}
\right.
\end{equation}
Definition \eqref{weights cycles kawa 1-d} is not ambiguous since the functions used are constant on the cycles and can be interpreted as functions of the cycle itself.
The flow is decomposed into a reversible part coming from the superposition of the reversible cycles and an irreversible one
$Q=Q^r+Q^i$.
The reversible part in the case of a jump of a particle from $x$ to $x+1$ is for example
\begin{equation}\label{qr}
Q^r(\eta,\eta^{x,x+1})=\tau_xw e^{-H^*_{\{x,x+1\}^c}}\eta(x)(1-\eta(x+1))\,,
\end{equation}
and similarly for a jump from $x+1$ to $x$.
Since any symmetric flow can be decomposed using these elementary cycles we obtain the most general form of the reversible rates having \eqref{gibbs} has an invariant measure
\begin{equation}\label{genrev}
c^r_{e^-,e^+}(\eta)=\tau_{\mathfrak e}w e^{H_{\mathfrak e}}\eta\left(e^-\right)\left(1-\eta\left(e^+\right)\right)\,.
\end{equation}

For the irreversible part $Q^i(\eta,\eta^{x.x+1})$  we may have a contribution from one cycle in $\mathcal C^\pm_{\{x-1.x\}}$ and one cycle in $\mathcal C^\pm_{\{x+1.x+2\}}$. Let us call $\chi_i\,, i=1,\dots , 4$ the characteristic functions associated with the local distribution of particles like in Figure \ref{fig: due v- due p Kaw1d} when $\mathfrak e=\{0,1\}$ and numbering from the top of the Figure toward the bottom. For example we have
$$
\chi_1(\eta)=\eta(-1)(1-\eta(0))\eta(1)(1-\eta(2))\,,
$$
for the first configuration from the top in Figure 3 and similar expressions for the other cases. We have
\begin{eqnarray}\label{qi0q}
Q^i(\eta,\eta^{x,x+1})&=&\tau_{x-1}\left(w^+\chi_2e^{-H^*_{W^c(\{0,1\})}}\right)+\tau_{x+1}\left(w^+\chi_1
e^{-H^*_{W^c(\{0,1\})}}\right)
\nonumber \\
&+&\tau_{x-1}\left(w^-\chi_1e^{-H^*_{W^c(\{0,1\})}}\right)+\tau_{x+1}\left(w^-\chi_4
e^{-H^*_{W^c(\{0,1\})}}\right)\,.
\end{eqnarray}
Likewise for jumps in the opposite direction we have
\begin{eqnarray}\label{qi0-q}
Q^i(\eta,\eta^{x+1,x})&=&\tau_{x-1}\left(w^+\chi_4e^{-H^*_{W^c(\{0,1\})}}\right)+\tau_{x+1}\left(w^+\chi_3
e^{-H^*_{W^c(\{0,1\})}}\right)
\nonumber \\
&+&\tau_{x-1}\left(w^-\chi_3e^{-H^*_{W^c(\{0,1\})}}\right)+\tau_{x+1}\left(w^-\chi_2
e^{-H^*_{W^c(\{0,1\})}}\right)\,.
\end{eqnarray}
Using the general rule \eqref{qpi} we obtain the rates of transition for the irreversible part
\begin{equation}\label{qi0}
c^i_{x,x+1}(\eta)=
\tau_{x-1}\Big[\left(w^+\chi_2+w^-\chi_1\right)e^{H_{W(\{0,1\})}}\Big]+\tau_{x+1}\Big[\left(w^+\chi_1+w^-\chi_4\right)
e^{H_{W(\{0,1\})}}\Big]\,.
\end{equation}
Likewise for jumps in the opposite direction we have
\begin{equation}\label{qi0-}
c^i_{x+1,x}(\eta)=\tau_{x-1}\Big[\left(w^+\chi_4+w^-\chi_3\right)e^{H_{W(\{0,1\})}}\Big]+\tau_{x+1}\Big[\left(w^+\chi_3+w^-\chi_2\right)
e^{H_{W(\{0,1\})}}\Big]\,.
\end{equation}
With some algebra putting all together we obtain the rates \eqref{euna},\eqref{edue}.

\subsection{Glauber dynamics}
We discuss in detail the one dimensional case. The higher dimensional cases can be discussed very similarly.
It is convenient to write the rates as
\begin{equation}\label{GG}
c_x(\eta)=\eta(x)c^-_x(\eta)+(1-\eta(x))c^+_x(\eta)\,.
\end{equation}
Decomposition \eqref{GG} identifies uniquely $c_x^\pm$ only if we require $D(c^\pm_x)\cap \{x\}=\emptyset$.

We search for rates having \eqref{gibbs} as invariant measure. We obtain that $c_x^\pm$ can be written as a sum of an irreversible part $c_x^{\pm,i}$ and of a reversible one $c_x^{\pm, r}$.
We have for the reversible part
\begin{equation}\label{arrit}
c_x^{r}=\tau_x\left(w e^{H_0}\right)\,,
\end{equation}
where $w$ is any non-negative function
such that $D(w)\cap \{0\}=\emptyset$. This is the form of the reversible rates in full generality.
The Bernoulli case \cite{GJLV} is recovered as a special case.
For the irreversible part we have to distinguish when $\eta(x)=1$ and when $\eta(x)=0$ using respectively
the following formulas \eqref{g10} and \eqref{g01}. We obtain
\begin{eqnarray}\label{g01r}
c_x^{\pm,i}=&=&\tau_x\Big[e^{H_{\{0,1\}}}\left(w_\mp\eta(1)+w_\pm(1-\eta(1))\right)\Big]\nonumber \\
&+& \tau_{x-1}\Big[e^{H_{\{0,1\}}}\left(w_\pm\eta(0)+w_\mp(1-\eta(0))\right)\Big]\,,
\end{eqnarray}
where $w^\pm$ are non-negative functions
such that $D(w^\pm)\cap \{0,1\}=\emptyset$.

\begin{remark}
If we fix the invariant measure as a Bernoulli measure of parameter $p$  the rates always satisfy the condition of macroscopic reversibility in \cite{GJLV}
\begin{equation} \label{macroscopic rev cond}
\frac{\bb{E}_\lambda \left( c^-_x \right)}{\bb{E}_\lambda\left( c^+_x\right)}=\frac{1-p}{p}\,, \qquad \forall \lambda\in[0,1]\,,
\end{equation}
for the hydrodynamic scaling limit of a Glauber+Kawasaki dynamics. In \eqref{macroscopic rev cond} $\mathbb E_\lambda$ denotes the expected value with respect to a Bernoulli measure of parameter $\lambda$ and $c^\pm_x$ are the one uniquely identified by \eqref{GG}
and the condition on the domain.
\end{remark}

The proof is as follows.
The reversible cycles of length 2 are of the type $C=(\eta,\eta^x,\eta)$.
We consider also minimal irreversible cycles involving two neighbouring sites $x$ and $x+1$ (see Figure \ref{fig: cicli Gl1-d}).
\begin{figure}[]
\[
\entrymodifiers={+<0.5ex>[o][F*:black]}
\xymatrix@C=1.5cm@R=1.cm{
*+[o][F*:black]{}\ar@{.}[r]^{\textit{\normalsize $\mathfrak e$}}& \ar@{}[r]|{\textit{\large{$\rightarrow$}}} &*+[o][F*:black]{}\ar@{.}[r]^{\textit{\normalsize $\mathfrak e$}}&*+[o][F*:black]{} \ar@{}[r]|{\textit{\large{$\rightarrow$}}} &
  \ar@{.}[r]^{\textit{\normalsize $\mathfrak e$}}&*+[o][F*:black]{} \ar@{}[r]|{\textit{\large{$\rightarrow$}}} &
\ar@{.}[r]^{\textit{\normalsize $\mathfrak e$}}&}
 \]
\caption{\small{An elementary irreversible cycle associated with the bond $\mathfrak e=\{x,x+1\}$. The configuration outside the window is frozen. The sequence following the arrows from the left to the right define a cycle in $\mathcal C^+_{\mathfrak e}$ and a cycle in $\mathcal C_{\mathfrak e}^-$ is generated going in the opposite direction.}}
\label{fig: cicli Gl1-d}
\end{figure}
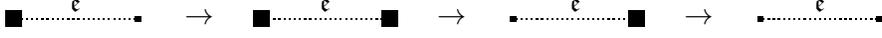
We fix the weights associated with the reversible cycles as $\rho(C):=e^{-H^*_{x^c}}\tau_xw$, where $w$ is any non-negative function
such that $D(w)\cap \{0\}=\emptyset$.
To fix the weights associated with the irreversible cycles consider $w^\pm$ non-negative functions
such that $D(w^\pm)\cap \{0,1\}=\emptyset$. We set
\begin{equation}\label{weights cycles Glauber 1-d}
\rho(C):=\left\{
\begin{array}{ll}
e^{-H^*_{\mathfrak e^c}}\tau_{\mathfrak e}w^+ & \textrm{if} \ C\in \mathcal C^+_{\mathfrak e}\,,\\
e^{-H^*_{\mathfrak e^c}}\tau_{\mathfrak e}w^- & \textrm{if} \ C\in \mathcal C^-_{\mathfrak e}\,.
\end{array}
\right.
\end{equation}
To the value of $Q(\eta,\eta^x)$ may contribute 4 irreversible cycles and a reversible one.
Two irreversible cycles are associated with the bond $\{x-1,x\}$ and two to the bond $\{x,x+1\}$.
When $\eta(x)=1$ we obtain
\begin{eqnarray}\label{g10}
Q^i(\eta,\eta^x)&=&\tau_x\left[e^{-H^*_{\{0,1\}^c}}\left(w^+\eta(1)+w^-(1-\eta(1)\right)\right]\nonumber \\
&+& \tau_{x-1}\left[e^{-H^*_{\{0,1\}^c}}\left(w^-\eta(0)+w^+(1-\eta(0)\right)\right]\,.
\end{eqnarray}
When $\eta(x)=0$ we obtain
\begin{eqnarray}\label{g01}
Q^i(\eta,\eta^x)&=&\tau_x\left[e^{-H^*_{\{0,1\}^c}}\left(w^-\eta(1)+w^+(1-\eta(1)\right)\right]\nonumber \\
&+& \tau_{x-1}\left[e^{-H^*_{\{0,1\}^c}}\left(w^+\eta(0)+w^-(1-\eta(0)\right)\right]\,.
\end{eqnarray}
For the reversible part we have in any case
\begin{equation}\label{gr}
Q^r(\eta,\eta^x)=\tau_x\left(we^{-H^*_{0^c}}\right)\,.
\end{equation}
Using \eqref{qpi} we obtain the rates introduced at the beginning.
A similar construction can be done also in dimension $d>1$.

\subsection{The general structure}
The constructions developed for some specific models can be summarized in a general form as follows.
Let $W\subseteq \mathbb Z_N^2$ be a finite region of the lattice. Consider $C$ a finite cycle on the configurations space restricted
to the finite region. This means a graph having vertices $\Gamma^W$ and edges between pairs of configurations that can be obtained one from the other according to a local modification depending on the type of dynamics we are considering, for example by a jump of one particle from one site to a nearest neighbor. We have $C=(\eta^{(1)}_W,\dots,\eta^{(k)}_W)$ where we write explicitly that the configurations are restricted to $W$. We can associate to the single cycle $C$ several cycles on the full configuration space $\Gamma^{\mathbb Z_N^2}$. In particular for any configuration of particles $\xi_{W^c}$ outside of the region $W$, we have the cycle
$C[\xi]=\big(\xi_{W^c}\eta^{(1)}_W,\dots , \xi_{W^c}\eta^{(k)}_W\big)$. Since the cycle $C[\xi]$ is labeled by the configuration $\xi_{W^c}$ we can associate a weight $\rho(C)=g(\xi)$ where $D(g)\cap W=\emptyset$.
Starting from these cycles we can generate also other cycles by translations
$\tau_xC[\xi]=\big(\tau_x\left[\xi_{W^c}\eta^{(1)}_W\right],\dots , \tau_x\left[\xi_{W^c}\eta^{(k)}_W\right]\big)$ and to have translational
covariant models we give the weights in such a way that $g(\xi)=\rho(C[\xi])=\rho(\tau_xC[\xi])$. We obtain a divergence free flow as a superposition of all these cycles
\begin{equation}\label{acor}
Q=\sum_{x\in \mathbb Z_N^2}\sum_{\xi_{W^c}}g(\xi)Q_{\tau_xC[\xi]}\,.
\end{equation}
More generally we can consider a collection $C_1,\dots ,C_l$ of finite cycles on configurations restricted to finite regions $W_1,\dots ,W_l$ with weights associated with functions $g_1,\dots,g_l$ with suitable domains. Accordingly in \eqref{acor} there will be also a sum over the possible types of cycles. We consider for simplicity the case of one single cycle like in \eqref{acor}. Suppose that we need to determine the value of $Q(\eta,\eta^{x,x+e^{(1)}})$ corresponding to a jump of one particle from $x$ to $x+e^{(1)}$ in the configuration $\eta$. Let $\chi_1,\dots ,\chi_k$ be the characteristic functions associated with the configurations restricted to $W$ along the cycle, i.e.
$$
\chi_i(\eta):=\left\{
\begin{array}{ll}
1 & \textrm{if}\ \eta_W=\eta^{(i)}_W\,, \\
0 &  \textrm{otherwise}\,.
\end{array}
\right.
$$
Suppose that along the cycle $C$ there are jumps of particles along the direction $e^{(1)}$ at times $i_1,\dots ,i_n$ in the positions
$x_{1},\dots,x_{n}\in W$. This means that $\eta_W^{\left(i_{m}+1\right)}=\left(\eta_W^{\left(i_m\right)}\right)^{x_m,x_m+e^{(1)}}$, $m=1,\dots ,n$. We then have
\begin{equation}\label{prto}
Q(\eta,\eta^{x,x+e^{(1)}})=\sum_{m=1}^n\tau_{x-x_{m}}\left[\chi_{i_m}(\eta)g(\eta)\right]\,.
\end{equation}
Having a fixed invariant measure $\pi$, the rates can be determined using \eqref{qpi}. The functions $g$ can be chosen in such a way that the rates are local.

\subsection{Examples with global cycles}\label{ewgc}

There are simple and natural models that cannot be constructed using local cycles. This is immediately clear for example for the totally asymmetric simple exclusion process (TASEP) on a ring, since particles can move only in one direction. Consider for example the case when particles can move only anticlockwise. In this case an elementary cycle is obtained as follows. Starting from each configuration move the particles one by one anticlockwise in such a way that each particle will occupy at the end the position
of the first particle in front of it in the anticlockwise direction. Note that there are many possible cycles of this type depending on the order on which particles are moved. Across each edge of the ring there will be in any case only one particle jumping. It is not difficult to see by \eqref{qpi} that using these cycles it is possible to construct TASEP with Bernoulli invariant measures.

Another example of this type is  a continuous time version of the irreversible Glauber
dynamics in \cite{PSS} for spins taking values $\pm 1$ on a ring. In this case elementary cycles can be naturally constructed in this way. We consider a given configuration and we flip the spins one by one starting from one single site and moving anticlockwise.
Going twice around the ring  we come back to the original configuration.

Both models can be understood also by the following symmetry argument.
Suppose that on a graph $(V,E)$ we have a flow $Q$ for which it is possible to find a bijection such that to any $(x,y)\in E$ we can associate a $(y',x)\in E$ and $Q(x,y)=Q(y',x)$. In this case the flow $Q$ is automatically divergence free since on each vertex the outgoing flow coincides with the incoming one. In the case of the TASEP from each configuration $\eta$, in the configuration space, the number of arrows exiting is equal to the number of arrows entering and coincides with the number of clusters on $\eta$. Giving the same weight to all the arrows corresponding to transitions for configurations having the same number of particles the bijection is automatically constructed. A similar more tricky construction can be done also for the model in \cite{PSS}.

\section{A functional discrete Hodge decomposition}
\label{fdhhd}
We discuss in this section our second discrete geometric construction. This is inspired by a construction in the theory of hydrodynamic limits that we briefly outline \cite{KL99,Spohn}.
The instantaneous current for a stochastic lattice gas evolving with a conservative dynamics is defined by
\begin{equation}\label{istcur}
j_\eta(x,y):=c_{x,y}(\eta)-c_{y,x}(\eta)\,.
\end{equation}
For each fixed configuration $\eta$ this is a discrete vector field. The intuitive interpretation of the instantaneous current is the rate at which particles cross the bond $(x,y)$. Let $\mathcal N_{x,y}(t)$ be the number of particles that jumped from site $x$ to site $y$ up to time $t$. The current flowed across the bond $(x,y)$ up to time $t$ is defined as
\begin{equation}\label{currver}
J_{x,y}(t):=\mathcal N_{x,y}(t)-\mathcal N_{y,x}(t)\,.
\end{equation}
This is a discrete vector field depending on the trajectory of the system of particles. The importance of the instantaneous current
is based on the key observation (see for example \cite{Spohn} section II 2.3) that
\begin{equation}\label{mart}
J_{x,y}(t)-\int_0^tj_{\eta(s)}(x,y)ds
\end{equation}
is a martingale. This is the main reason why the instantaneous current is relevant in understanding
the collective behavior of particle systems. We will not use this fact in the following.

A relevant notion in the derivation of the hydrodynamic behavior for diffusive particle systems is the definition of gradient particle system. The basic definition is the following. A particle system is called of \emph{gradient type} if there exists a local function $h$ such that
\begin{equation}\label{grgr}
j_\eta(x,y)=\tau_yh(\eta)-\tau_x h(\eta)\,.
\end{equation}
The relevance of this notion is on the fact that the proof of the hydrodynamic limit for gradient systems is extremely simplified.
Moreover for gradient and reversible models it is possible to obtain explicit expressions of the transport coefficients.

We develop a generalization of the Hodge decomposition discussed in Section \ref{disch}. This applies to discrete vector fields suitably depending  on the configurations of particles and vector fields of the form \eqref{grgr} play the role of the gradient vector fields. Circulations will also be suitably defined. We prove a functional Hodge decomposition of
translational covariant discrete vector fields. This means vector fields depending on the configuration $\eta$ and such that
\begin{equation}\label{tc}
j_\eta(x,y)=j_{\tau_z\eta}(x+z,y+z)\,.
\end{equation}
The instantaneous current of any translational covariant stochastic lattice gas is translational covariant.
Note however that our results hold for any translational covariant discrete vector fields. We will apply
this decomposition in Section \ref{SeO} to different vector fields suitably associated to stochastic lattice gases.

\subsection{The one dimensional case}

We consider in this section the one dimensional torus with $N$ sites $\mathbb Z_N:=\mathbb Z/(N\mathbb Z)$ that we represent as a ring
such that we move anticlockwise going from $x$ to $x+1$. We have the following Theorem.

\begin{theorem}\label{belteo}
Let $j$ be a translational covariant discrete vector field. Then there exists a function $h$ and a translational invariant function $C$ such that
\begin{equation}\label{imbr}
j_\eta(x,x+1)=\tau_{x+1}h(\eta)-\tau_{x}h(\eta)+C(\eta)\,.
\end{equation}
The function $C$ is uniquely identified and coincides with
\begin{equation}\label{prim}
C(\eta)=\frac 1N\sum_{x\in \mathbb Z_N} j_\eta(x,x+1)\,.
\end{equation}
The function $h$ is uniquely identified up to an arbitrary additive translational invariant function and coincides with
\begin{equation}\label{formulah}
h(\eta)=\sum_{x=1}^{N-1}\frac{x}{N}j_\eta(x,x+1)\,.
\end{equation}
\end{theorem}

\begin{proof}
A short proof (apart the uniqueness statements that have to be proved separately like showed at the end) can be obtained simply inserting \eqref{prim} and \eqref{formulah} into \eqref{imbr}. We give a proof that shows how it is
possible to guess the explicit expression \eqref{formulah} and gives some intuition. The basic idea is to construct several potentials for $j_\eta-C(\eta)$ one for each lattice site of the ring and then to consider the spatial average of all of them. With a suitable
use of the translational covariance we obtain in this way the function $h$.

First of all summing over $x\in \mathbb Z_N$ both sides of \eqref{imbr} we obtain that $C$ is uniquely determined by \eqref{prim}.
By \eqref{tc} we have that $C$ is invariant by translations
$$
C(\tau_z\eta)=\frac 1N\sum_{x=1}^Nj_{\tau_z\eta}(x,x+1)=\frac 1N\sum_{x=1}^Nj_\eta(x-z,x-z+1)=C(\eta)\,.
$$
We define a function $U:\mathbb Z_N\times \Gamma^{\mathbb Z_N}\times \mathbb Z_N\to \mathbb R^+$ that associate to the triple $(x,\eta,y)$ the value $U_x(\eta,y)$ defined by $U_x(\eta,x):=0$ and for $y\neq x$ by
\begin{equation}\label{sete}
U_x(\eta,y):=\sum_{z\in [x,y-1]}\left[j_\eta(z,z+1)-C(\eta)\right]\,,
\end{equation}
where $[x,y-1]$ is the collection of sites encountered going from $x$ to $y-1$ anticlockwise.
Since
\begin{equation}\label{egrad}
\sum_{z=1}^{N}\left[j_\eta(z,z+1)-C(\eta)\right]=0
\end{equation}
the function $U$ is well defined. Condition \eqref{egrad} coincides with the condition that $j_\eta-C(\eta)$ is a gradient vector field for any configuration $\eta$ and $U_x$ is one of the associated potentials fixed in such a way that it is zero at $x$. We have indeed
\begin{equation}\label{indeed}
U_x(\eta,y+1)-U_x(\eta,y)=j_\eta(y,y+1)-C(\eta)\,.
\end{equation}
Consider the averaged function
\begin{equation}\label{indeedav}
\overline U(\eta,y):=\frac 1N \sum_{x=1}^N U_x(\eta,y)\,,
\end{equation}
for which still we have the equality
\begin{equation}\label{indeedproprio}
\overline U(\eta,y+1)-\overline U(\eta,y)=j_\eta(y,y+1)-C(\eta).
\end{equation}
We will show that there exists a function $h: \Gamma^{\mathbb Z_N}\to \mathbb R$ such that $\overline U(\eta,y)=\tau_yh(\eta)$ and \eqref{indeedproprio} becomes \eqref{imbr}. First of all we prove the invariance property
\begin{equation}\label{invariance}
\overline U(\tau_z \eta, y+z)=\overline U(\eta,y)\,.
\end{equation}
We have
\begin{eqnarray*}
\overline U(\tau_z \eta, y+z)&=&\frac 1N\sum_{x=1}^NU_x (\tau_z\eta,y+z)\\
&=&\frac 1N \sum_{x=1}^N\sum_{v\in [x,y+z-1]}  \left[j_{\tau_z\eta}(v,v+1)-C(\tau_z\eta)\right]\\
&=&\frac 1N \sum_{x=1}^{N}\sum_{v\in [x,y+z-1]}\left[j_\eta(v-z,v-z+1)-C(\eta)\right]\\
&=&\frac 1N \sum_{x=1}^{N}U_{x-z}(\eta,y)=
\overline U(\eta,y)\,.
\end{eqnarray*}
We used the invariance \eqref{tc} and the translational invariance of $C$.
If we define $h(\eta):=\overline U(\eta,0)$, by the invariance \eqref{invariance} we have
\begin{equation}\label{lem4}
\overline U(\eta,y)=\overline U(\tau_{-y}\eta,0)=h(\tau_{-y}\eta)=\tau_yh(\eta)\,.
\end{equation}
Formula \eqref{formulah} can be easily deduced from \eqref{indeedav} and the definition of $h$ or directly checking the validity of \eqref{imbr} using \eqref{prim}.
To finish the proof suppose now that there are two functions $h$ and $h'$ satisfying \eqref{imbr}. Since $C$ is univocally determined we deduce $\tau_{x+1}(h-h')=\tau_x(h-h')$ that means that $h-h'$ is translational invariant. This finishes the proof of the Theorem.
\end{proof}
The basic idea of the above Theorem is the usual strategy to construct the potential of a gradient discrete vector field plus a subtle use of the translational covariance of the model.
It is interesting to observe that a one dimensional system of particles is of gradient type (with a possibly not local $h$)
if and only if $C(\eta)=0$. This corresponds to say that for any fixed configuration $\eta$ then $j_\eta(x,y)$
is a gradient vector field. This was already observed in \cite{BDGJLstoc-int,N}.
We discuss some examples

\subsubsection{The gradient case}\label{surp}
If the vector field is of gradient type, i.e. $j_\eta(x,x+1)=\tau_{x+1}h(\eta)-\tau_xh(\eta)$ then we have
$$
\sum_{x=1}^{N-1}\frac{x}{N}j_\eta(x,x+1)=h(\eta)-\frac 1N\sum_{x=1}^N\tau_xh(\eta)\,,
$$
and since the second term on the right hand side is translational invariant we reobtain in formula \eqref{formulah} the original function $h$ up to a translational invariant function.

\subsubsection{The 2-SEP}
The model we are considering is the 2-SEP for which on each lattice site there can be at most 2 particles (\cite{KL99} section 2.4). The generalization to k-SEP can be easily done likewise. More precisely we have $\Gamma=\{0,1,2\}$.
The dynamics is defined by the rates
$$
c_{x,x\pm 1}(\eta)=\chi^+(\eta(x))\chi^-(\eta(x\pm 1))
$$
where $\chi^+(\alpha)=1$ if $\alpha >0$ and zero otherwise while $\chi^-(\alpha)=1$ if $\alpha<2$ and zero otherwise.
We denote also $D^\pm_\eta(x,x+1)$ local functions associated with the presence on the bond $(x,x+1)$ of what we call respectively a positive or negative discrepancy. More precisely $D^+_\eta(x,x+1)=1$ if $\eta(x)=2$ and $\eta(x+1)=1$ and zero otherwise. We have instead $D^-_\eta(x,x+1)=1$ if $\eta(x+1)=2$ and $\eta(x)=1$ and zero otherwise. We define also $D_\eta:=D^+_\eta-D^-_\eta$.

The instantaneous current across the edge $(x,x+1)$ associated with the configuration $\eta$ is
$$
j_\eta(x,x+1):=\chi^+(\eta(x))-\chi^+(\eta(x+1))+D_\eta(x,x+1)
$$
For this specific model formulas \eqref{prim} and \eqref{formulah} become
$$
C(\eta)=\frac 1N\sum_{x\in \mathbb Z_N}D_\eta(x,x+1)
$$
$$
h(\eta)=-\chi^+(\eta(0))+\sum_{x=1}^{N-1}\frac{x}{N}D_\eta(x,x+1)\,.
$$

\subsubsection{ASEP}
The asymmetric simple exclusion process is characterized by the rates $c_{x,x+1}(\eta)=p\eta(x)(1-\eta(x+1))$ and $c_{x,x-1}(\eta)=q\eta(x)(1-\eta(x-1))$. Recall that $\mathfrak {C}(\eta)$ is the collection of clusters of particles in the configuration $\eta$. Given a cluster $c\in \mathfrak C$ we call $\partial^lc,\partial^rc\in\{0,1,\dots ,N-1\}$ the first element of the lattice on the left of the leftmost site of the cluster and the rightmost site of the cluster respectively. The decomposition \eqref{imbr} holds with
\begin{equation}\label{caldaia}
C(\eta)=\frac{(p-q)\left|\mathfrak C(\eta)\right|}{N}\,,
\end{equation}
where $\left|\mathfrak C(\eta)\right|$ denotes the number of clusters and
\begin{equation}\label{bagnetto}
h(\eta)=\frac 1N\sum_{c\in \mathfrak C(\eta)}\left[p\partial^rc-q\partial^lc\right]\,.
\end{equation}

\subsubsection{Reversible and gradient models}
We show in a concrete example the usefulness of the criterium to recognize a gradient model
illustrated just before Section \ref{surp}.
A problem of interest is to find models that are at the same time of gradient type and reversible. This problem is discussed in Section II.2.4 of \cite{Spohn}. In dimension $d>2$ it is difficult to find models satisfying the two conditions at the same time. In $d=1$ it has been proved in \cite{N} that this is always possible for any invariant finite range  Gibbs measure and nearest neighbors exchange dynamics with $\Gamma=\{0,1\}$. Let us show how it is possible to prove this fact with a simple argument for a special class of interactions. The interactions that we consider are of the form $J_A\prod_{x\in A}\eta(x)$ where $A$ are intervals of the lattice. The numbers $J_A$ are translational invariant and satisfy the relation $J_A=J_{\tau_xA}$. Consider the most general form of the reversible rates \eqref{genrev} and fix the arbitrary function as $w=1$. The instantaneous current is given by $j_\eta(x,x+1)=e^{H_{\{x,x+1\}}}\left[\eta(x)-\eta(x+1)\right]$. The instantaneous current is different from zero only on the left and right boundary of each cluster of particles. More precisely it is positively directed on the right boundary and negatively directed on the left one. Since the interactions are associated only to intervals and are translational invariant then $e^{H_{\mathfrak e}}$ when $\mathfrak e$ is a boundary edge of a cluster depends only on the size of the cluster. This means that for each cluster the sum of the instantaneous currents on the two boundary edges is identically zero and this implies that
$\sum_{x\in \mathbb Z_N}j_\eta(x,x+1)=0$ for any configuration $\eta$. By Theorem \ref{belteo}, this coincides with the gradient condition. We generated in this way
a gradient reversible dynamics for each Gibbs measure of this type.

\subsection{The two dimensional case}

We consider now the two dimensional torus with vertices $\mathbb Z^2_N:=\mathbb Z^2/N\mathbb Z^2$ and the set of edges $\mathcal E$  given by the pairs of vertices $\{x,y\}$ that are at distance 1.
The functional discrete Hodge decomposition in the 2 dimensional case is as follows.

\begin{theorem}\label{belteo2}
Let $j$  be a translational covariant discrete vector field. Then there exist 4 functions $h,g,C^{(1)}, C^{(2)}$ on configurations of particles such that for an edge of the type $e=(x,x\pm e^{(i)})$ we have
\begin{equation}\label{hodgefun2}
j_\eta(e)=\big[\tau_{e^+}h(\eta)-\tau_{e^-}h(\eta)\big]+\big[\tau_{\mathfrak f^+(e)}g(\eta)-\tau_{\mathfrak f^-(e)}g(\eta)\big]\pm C^{(i)}(\eta)\,.
\end{equation}
The functions $C^{(i)}$ are translational invariant and uniquely identified. The functions $h$ and $g$ are uniquely identified up to additive arbitrary translational invariant functions.
\end{theorem}
\begin{proof}
Since for any fixed $\eta$ we have that $j_\eta(\cdot)$ is an element of $\Lambda^1$ we can perform for any fixed $\eta$ the Hodge decomposition writing $j_\eta=j_\eta^\nabla+j_\eta^\delta+j_\eta^H$. The harmonic component $j_\eta^H$ is related to the functions
$C^{(i)}$ that are equivalent to the constants $c_i$ computed by \eqref{cco}. Since we have now discrete vector fields depending on the configuration $\eta$ we will obtain not constants but functions of the configurations. In particular we have a formula completely analogous of \eqref{cco} that is
\begin{equation}\label{defci}
C^{(i)}(\eta):=\frac{1}{N^2}\sum_{x\in \mathbb Z^2_N}j_\eta(x,x+e^{(i)})\,.
\end{equation}
The functions in \eqref{defci} are clearly translational invariant and uniquely determined as can be seen taking
the scalar product with $\phi^{(i)}$ on both sides of \eqref{hodgefun2}. We have
$$
j^H_\eta=C^{(1)}(\eta)\phi^{(1)}+C^{(2)}(\eta)\phi^{(2)}\,.
$$
As a second step we show that the translational covariance of $j_\eta$ is inherited also by the other components $j^\delta_\eta, j^\nabla_\eta$.  Fix arbitrarily $z\in \mathbb Z_N^2$ and consider two new discrete vector fields depending on configurations
\begin{equation}\label{invdgH}
\left\{
\begin{array}{l}
\tilde j_\eta^\delta(x,y) := j_\eta^\delta(x+z,y+z)\,, \\
\tilde j_\eta^\nabla(x,y) := j_\eta^\nabla(x+z,y+z)\,. \\
\end{array}
\right.
\end{equation}
For any fixed $\eta$ we have $\tilde j^\delta_\eta\in \delta\Lambda^2$ and $\tilde j^\nabla_\eta\in \nabla\Lambda^0$ since translations preserve these properties. We can write
\begin{eqnarray*}
& &\tilde j_\eta^\delta(x,y) +\tilde j_\eta^\nabla(x,y)+j^H_\eta(x,y)=j_\eta(x+z,y+z)=\\
& &j_{\tau_{-z}\eta}(x,y)=j_{\tau_{-z}\eta}^\delta(x,y)+j_{\tau_{-z}\eta}^\nabla(x,y)+j_\eta^H(x,y)\,.
\end{eqnarray*}
We used the translational covariance of $j_\eta$ and the translational invariance of $j^H_\eta$.
Since for any fixed $\eta$ the Hodge decomposition is unique we obtain $j_{\tau_{z}\eta}^\delta(x+z,y+z)=j_{\eta}^\delta(x,y)$ and
$j_{\tau_{z}\eta}^\nabla(x+z,y+z)=j_{\eta}^\nabla(x,y)$ that is the translational covariance of $j_\eta^\delta$ and $j_\eta^\nabla$.

Since for any fixed $\eta$ we have that $j_\eta^\nabla$ is translational covariant and of gradient type, we  show that there exists a function $h$ such that
\begin{equation}\label{e2}
j_\eta^\nabla(x,x+e^{(i)})=\tau_{x+e^{(i)}}h(\eta)-\tau_x h(\eta) \qquad i=1,2\,.
\end{equation}
The proof is very similar to the one dimensional case. We define a function $U_x(\eta, y)$ where $x,y\in \mathbb Z_N^2$. This is defined by
\begin{equation}\label{grd2}
U_x(\eta,y):=\sum_{i=0}^{k-1}j_\eta^\nabla(z^{(i)},z^{(i+1)})\,,
\end{equation}
where $z^{(0)}, \dots z^{(k)}$ is any path on $\mathbb Z_N^2$ going from $z^{(0)}=x$ to $z^{(k)}=y$. Since $j^\nabla$ is of gradient type the value is independent of the path.
As in the one dimensional case we consider the averaged function
\begin{equation}\label{indeedav2}
\overline U(\eta,y):=\frac{1}{N^2} \sum_{x\in \mathbb Z_N^2} U_x(\eta,y)\,,
\end{equation}
for which we have the equality
\begin{equation}\label{indeedproprio2}
\overline U(\eta,y)-\overline U(\eta,x)=j_\eta^\nabla(x,y)\,,
\end{equation}
for any $(x,y)\in E$.
To finish the proof we need to show that there exists a function $h:\Gamma^{\mathbb Z_N^2}\to \mathbb R$ such that $\overline U(\eta,y)=\tau_yh(\eta)$ and \eqref{indeedproprio2} becomes \eqref{e2}. First of all we prove the invariance property
\begin{equation}\label{invariance2}
\overline U(\tau_z \eta, y+z)=\overline U(\eta,y)\,.
\end{equation}
This follows by the symmetry
\begin{equation}\label{kgl}
U_x(\eta,y)=U_{x+z}(\tau_z\eta, y+z)
\end{equation}
that is obtained directly by the definition and the translational covariance of $j^\nabla$.
Likewise in the one dimensional case we have then
\begin{eqnarray*}
& &\overline U(\tau_z \eta, y+z)=\frac{1}{N^2}\sum_{x\in \mathbb Z_N^2}U_x (\tau_z\eta,y+z)\\
& &=\frac{1}{N^2} \sum_{x'\in \mathbb Z_N^2}U_{x'+z} (\tau_z\eta,y+z)=\frac{1}{N^2} \sum_{x'\in \mathbb Z_N^2}U_{x'} (\eta,y)=\overline U(\eta,y)\,.
\end{eqnarray*}
We used the change of variables $x'=x-z$ and the invariance \eqref{kgl}.
If we define $h(\eta):=\overline U(\eta,0)$, by the invariance \eqref{invariance2} we have
\begin{equation}\label{lem4*}
\overline U(\eta,y)=\overline U(\tau_{-y}\eta,0)=h(\tau_{-y}\eta)=\tau_yh(\eta)\,.
\end{equation}
The uniqueness of $h$ up to an additive translational invariant function is obtained like in the one dimensional case.

Consider now $j_\eta^\delta\in \delta\Lambda^2$ for any $\eta$ and translational covariant. This means that for any $\eta$ there exists a 2 form $\psi(\eta;f)$ such that
\begin{equation}\label{craxi}
j_\eta^\delta(e)=\psi\big(\eta; f^+(e)\big)-\psi\big(\eta; f^-(e)\big)\,, \qquad e\in E\,,
\end{equation}
where we recall that $f^+(e)$ is the unique anticlockwise oriented face to which $e$ belongs and $f^-(e)$ is the unique anticlockwise oriented face to which $-e$ belongs. A two form satisfying \eqref{craxi} is identified up to an arbitrary constant. For any $x^*\in \mathbb Z_N^{*2}$  we introduce $\psi_{x^*}(\eta; f)$ that is the 2 form satisfying \eqref{craxi} and
such that $\psi_{x^*}(\eta;f^*)=0$, where $f^*$ is the anticlockwise oriented face associated with $x^*$. We define
\begin{equation}\label{betto}
\overline\psi(\eta; f):=\frac{1}{N^2}\sum_{x^*\in \mathbb Z_N^{*2}}\psi_{x^*}(\eta; f)
\end{equation}
that satisfies
\begin{equation}\label{craxi2}
j_\eta^\delta(e)=\overline\psi\big(\eta; f^+(e)\big)-\overline\psi\big(\eta; f^-(e)\big)\,, \qquad e\in E\,.
\end{equation}
With essentially the same computation as in the gradient case we can show that
\begin{equation}\label{cius}
\overline\psi(\tau_z \eta; \tau_z f)=\overline\psi(\eta; f)\,, \qquad \forall z\,.
\end{equation}
We define the function $g(\eta):=\overline\psi(\eta,f^0)$ where $f^0$ is the anticlockwise oriented face associated with the vertex
$\left(\frac 12, \frac 12\right)\in \mathbb Z_N^{*2}$. From \eqref{cius} we obtain
\begin{equation}\label{gugarios}
\overline\psi(\eta; f)=\overline\psi(\tau_{-x}\eta; f^0)=g(\tau_{-x}\eta)=\tau_{\mathfrak f}g(\eta)
\end{equation}
where $x$ is the vertex of the original lattice and such that $x^*(f)=\left(\frac 12,\frac 12\right)+x$ and $x^*(f)$ is the vertex of the dual lattice associated with the face $f$.
Using \eqref{gugarios} the relation \eqref{craxi2} becomes
\begin{equation}
j_\eta^\delta(e)=\tau_{\mathfrak f^+(e)}g(\eta)-\tau_{\mathfrak f^-(e)} g(\eta)\,.
\end{equation}
Uniqueness of $g$ up to an arbitrary additive translation invariant function follows by the usual argument.
\end{proof}

\subsubsection{A non gradient lattice gas with local decomposition}
We construct a model of particles satisfying an exclusion rule, with jumps only trough nearest neighbors sites and having a non trivial decomposition of the instantaneous current \eqref{hodgefun2} with $C^{(i)}=0$ and $h$ and $g$ local functions. The functions $h$ and $g$ have to be chosen suitably in such a way that the instantaneous current is always zero inside cluster of particles and empty clusters and has to be always such that $j_\eta(x,y)\geq 0$ when $\eta(x)=1$ and $\eta(y)=0$. A possible choice is the following perturbation of the SEP. We fix $h(\eta)=-\eta(0)$ and $g(\eta)$ with $D(g)=\{0, e^{(1)},e^{(2)}, e^{(1)}+e^{(2)} \}$ (we denote by $0$ the vertex $(0,0)$) defined as follows. We have $g(\eta)=\alpha$ if $\eta(0)=\eta(e^{(1)}+e^{(2)})=1$ and
$\eta(e^{(1)})=\eta(e^{(2)})=0$. We have also  $g(\eta)=\beta$ if $\eta(0)=\eta(e^{(1)}+e^{(2)})=0$ and
$\eta(e^{(1)})=\eta(e^{(2)})=1$. The real numbers $\alpha,\beta$ are such that $|\alpha|+|\beta|<1$. For all the remaining configurations we have $g(\eta)=0$.  Since $\Gamma=\{0,1\}$ the rates of jump are uniquely determined by $c_{x,y}(\eta)=\left[j_\eta(x,y)\right]_+$.

\subsubsection{A perturbed zero range dynamics} We consider a dynamics similar to the dynamics
in Section \ref{basile} but with $\Gamma=\mathbb N$ and having a local non trivial decomposition. We say that the face
$\{0, e^{(1)}, e^{(2)}, e^{(1)}+e^{(2)}\}$ is full in the configuration
$\eta\in \mathbb N^{\mathbb Z_N^2}$ if $\min\{\eta(0), \eta(e^{(1)}), \eta(e^{(2)}), \eta(e^{(1)}+e^{(2)})\}>0$.
Consider two non negative functions $w^\pm$ that are identically zero when the face $\{0, e^{(1)}, e^{(2)}, e^{(1)}+e^{(2)}\}$
is not full. Given a positive function $\tilde h:\mathbb N\to \mathbb R^+$, we define the rates of jump as
\begin{equation}\label{chiu}
c_{e^-,e^+}(\eta)=\tilde h(\eta(e^-))+\tau_{\mathfrak f^+(e)}w^++\tau_{\mathfrak f^-(e)}w^-\,.
\end{equation}
This corresponds to a perturbation of a zero range dynamics such that one particle jumps from one site with $k$ particles with a rate $\tilde h(k)$. The perturbation increases the rates of jump if the jump is on the edge of a full face. The gain depends on the orientation and the effect of different faces is additive. For such a model the instantaneous current
has a local decomposition \eqref{hodgefun2} with $h(\eta)=-\tilde h(\eta(0))$ and $g(\eta)=w^+(\eta)-w^-(\eta)$.

\section{Stationarity and orthogonality}\label{SeO}

We discuss a different approach to the problem of constructing non reversible stationary non equilibrium states. This is based on the functional Hodge decomposition discussed in Section \ref{fdhhd}. The main idea is that it is possible to interpret the stationary equations as an orthogonality condition with respect to a suitable harmonic discrete vector field. Translational covariant discrete vector fields can be generated using the functional discrete Hodge decomposition. This interpretation can be given in general, however, to have a clearer view of the geometric construction, we discuss the simplified case
of a Bernoulli invariant measure. More general measures can be discussed as well (see Remark \ref{remarco}). We find families of solutions but generalizations are possible.

This approach and the one developed in Section \ref{imdff} are just two different perspectives
on the same problem and are therefore strongly related. In particular it should be possible
to prove a general Theorem stating the equivalence of the two conditions. Less ambitiously, given any model associated to an invariant measure, it should be possible to verify both the orthogonality and the divergence free condition. We are not going to discuss
this equivalence here. We just show in the next first example how to recognize similar structures in the models constructed.

\subsection{One-dimensional Kawasaki dynamics}\label{1dkb}
We start discussing possibly the simplest case, a one dimensional  Kawasaki dynamics with a translational covariant rate of exchange given by $c_{\mathfrak e}(\eta)$. Of course we can write
the exchange rate as the sum of two jumps rates
\begin{equation}\label{e=s}
c_{\mathfrak e}(\eta)=c_{e^-,e^+}(\eta)+c_{e^+,e^-}(\eta)\,.
\end{equation}
We consider the case of the one dimensional ring. We can write $c_{x,x+1}(\eta)=\eta(x)(1-\eta(x+1))\tilde{c}_{x,x+1}(\eta)$ and $c_{x+1,x}(\eta)=\eta(x+1)(1-\eta(x))\tilde{c}_{x+1,x}(\eta)$ where  $D\left(\tilde{c}_{x,y}\right)\cap\{x,y\}=\emptyset$. The stationary condition can be written as
\begin{equation}\label{mauf}
\mu(\eta)\sum_{\mathfrak e\in \mathcal E}\left[c_{\mathfrak e}(\eta)-\frac{\mu(\eta^{\mathfrak e})}{\mu(\eta)}c_{\mathfrak e}(\eta^{\mathfrak e})\right]=0\,.
\end{equation}
Using \eqref{e=s} and the fact that $\mu$ is a Bernoulli measure so that $\frac{\mu(\eta^{\mathfrak e})}{\mu(\eta)}=1$, with some algebra we can show that \eqref{mauf} is equivalent to
\begin{equation}\label{equivuf}
\sum_{x\in \mathbb Z_N}\Big[\left(\eta(x+1)-\eta(x)\right)\left(\tilde{c}_{x+1,x}(\eta)-\tilde{c}_{x,x+1}(\eta)\right)\Big]=0\,.
\end{equation}
The expression inside squared parenthesis in \eqref{equivuf} is symmetric for the exchange $x\leftrightarrow x+1$. We can naturally interpret \eqref{equivuf} as the orthogonality condition
\begin{equation}\label{orti}
\langle \gamma_\eta , \mathbb I\rangle=0\,,
\end{equation}
where $\mathbb I$ is the harmonic discrete vector field defined by $\mathbb I (x,x+1)=1$ (and consequently $\mathbb I(x+1,x)=-1$) while $\gamma_\eta$ is the discrete vector field defined by
$$\gamma_\eta(x,x+1)=\Big(\eta(x+1)-\eta(x)\Big)\Big(\tilde{c}_{x+1,x}(\eta)-\tilde{c}_{x,x+1}(\eta)\Big)$$
(and consequently $\gamma_\eta(x+1,x)=-\gamma_\eta(x,x+1)$). It may appear strange to construct an
antisymmetric object starting from a symmetric one (the argument of \eqref{equivuf}), but this is exactly what must
happens to interpret \eqref{equivuf} as a scalar product.
\begin{remark}\label{remarco}
In the case of a different invariant measure the stationary condition can be written again as \eqref{orti} where the discrete vector field $\gamma_\eta$ is given by
\begin{equation}\label{ggibbs}
\gamma_\eta(x,x+1):=\Big(\eta(x)- r_x(\eta)\eta(x+1)\Big)\left(\tilde c_{x,x+1}(\eta)-\frac{ \tilde c_{x+1,x}(\eta)}{r_x(\eta)}\right)\,,
\end{equation}
with $r_x(\eta):=\frac{\pi\left(\eta_{\{x,x+1\}^c}1_x0_{x+1}\right)}{\pi\left(\eta_{\{x,x+1\}^c}0_x1_{x+1}\right)}$. Also in this case $\gamma_\eta$ is translational covariant and assuming that $r_x(\eta)$ is local (as it is for a finite range Gibbs measure) we can proceed
similarly to the Bernoulli case.
\end{remark}
We use the orthogonality \eqref{orti} and the similar relationship in the following just algebraically.
It does not seem to be a direct physical interpretation of the vector fields involved.
The vector field $\gamma_\eta$ is translational covariant and can be decomposed like \eqref{imbr}. The orthogonality condition \eqref{orti} can be satisfied if and only if the harmonic part $C$ in the decomposition \eqref{imbr} is identically zero and we get that the stationary condition is satisfied if and only if there exists a function $h$ such that
\begin{equation}\label{non}
\gamma_\eta(x,x+1)=(\eta(x+1)-\eta(x))\left(\tilde{c}_{x+1,x}(\eta)-\tilde{c}_{x,x+1}(\eta)\right)=\tau_{x+1}h(\eta)-\tau_xh(\eta)\,.
\end{equation}
To solve the above equation we have to consider a function $h$ such that the right hand side of \eqref{non} is zero when
$\eta(x)=\eta(x+1)$ since the left hand side is clearly zero in this case. Moreover when $\eta(x)\neq \eta(x+1)$, we have to impose that
\begin{equation}\label{non+}
\frac{\tau_{x+1}h(\eta)-\tau_xh(\eta)}{\eta(x+1)-\eta(x)}
\end{equation}
is a function invariant under the exchange of the values of $\eta(x)$ and $\eta(x+1)$. This is because by \eqref{non} we have that \eqref{non+}
has to coincide with $\tilde{c}_{x+1,x}(\eta)-\tilde{c}_{x,x+1}(\eta)$ that does not depend on $\eta(x),\eta(x+1)$. If we fix an $h$
that satisfies these constraints, we obtain by \eqref{non} the rates $\tilde c$.

A first possibility is to fix $h(\eta)=\eta(0)C(\eta)$ where $C(\eta)$ is a translational invariant function.
In this way we get $\tau_{x+1}h(\eta)-\tau_xh(\eta)=(\eta(x+1)-\eta(x))C(\eta)$ and we have
$$
\tilde{c}_{x+1,x}(\eta)-\tilde{c}_{x,x+1}(\eta)=C(\eta)\,.
$$
Since the left hand side does not depend on $\eta(x)$ and $\eta(x+1)$ and $C$ is translational invariant, the only possibility is that
$C$ is a constant function.
All the non negative solutions $X,Y$ of the equation $X-Y=A$ are given by
\begin{equation}\label{posol}
\left\{
\begin{array}{l}
X=[A]_++S\,,\\
Y=[-A]_++S\,,
\end{array}
\right.
\end{equation}
where $[\cdot]_+$ denotes the positive part and $S\geq 0$ is arbitrary.
We obtain that the general solution in this case is
\begin{equation}\label{elnc}
\left\{
\begin{array}{l}
\tilde c_{x+1,x}(\eta)=\Big[C\Big]_++\tau_xs(\eta)\,, \\
\tilde c_{x,x+1}(\eta)=\Big[-C\Big]_++\tau_xs(\eta)\,,
\end{array}
\right.
\end{equation}
where $s$ is an arbitrary non negative function such that $D(s)\cap\{0,1\}=\emptyset$. In formula \eqref{elnc} the additive part involving the function $s$ corresponds to the reversible part while the remaining part corresponds to the irreversible part. The asymmetric exclusion process is obtained as a special case.
The corresponding decomposition into cycles contains necessarily
global cycles as discussed in Section \ref{ewgc}.

Another possibility is to consider a function $h$ of the form $h=\tau_1 g+g$. Then we have
\begin{equation}\label{grgh}
\tau_{x+1}h-\tau_xh=\tau_{x+2}g-\tau_xg\,.
\end{equation}
A very general class  of functions $h$ that can be used in \eqref{non} is then obtained by $h=g+\tau_1g$ where
\begin{equation}\label{trd}
g(\eta)=\big(\eta(-1)-\eta(-2)\big)\big(\eta(1)-\eta(0)\big)\tilde g(\eta)
\end{equation}
and $\tilde g$ is any function such that $D(\tilde g)\cap W(\{-1,0\})=\emptyset$, where we recall that the symbol $W(\mathfrak e)$ has been defined just above \eqref{weights cycles kawa 1-d}.

We obtain for our class of functions $h$ a general form of the rates given by
\begin{equation}\label{eln}
\left\{
\begin{array}{l}
\tilde c_{x+1,x}(\eta)=\Big[(\eta(x+3)-\eta(x+2))\tau_{x+2}\tilde g-(\eta(x-1)-\eta(x-2))\tau_x\tilde g\Big]_++\tau_xs(\eta)\,, \\
\tilde c_{x,x+1}(\eta)=\Big[(\eta(x-1)-\eta(x-2))\tau_x\tilde g-(\eta(x+3)-\eta(x+2))\tau_{x+2}\tilde g\Big]_++\tau_xs(\eta)\,,
\end{array}
\right.
\end{equation}
where $s$ is an arbitrary non negative function such that $D(s)\cap\{0,1\}=\emptyset$. In formula \eqref{eln} the additive part involving the function $s$ corresponds to the reversible part while the part involving the function $\tilde g$ corresponds to the irreversible part.

To compare the models obtained in \eqref{eln} with the models obtained in Section \ref{odkd} we have first of all
to recall that here we are considering product invariant measures so that in formulas \eqref{euna} and \eqref{edue}
the value of the energy is constant and can be incorporated into the arbitrary functions. If we fix $\tau_1 \tilde g=w^+-w^-$ we obtain that the rates obtained in Section \ref{odkd} are a subset of the models defined by \eqref{eln}. This is obtained observing that
in both cases we will have
\begin{equation}\label{pes}
\tilde{c}_{x+1,x}-\tilde{c}_{x,x+1}=\big(\eta(x+3)-\eta(x+2)\big)\tau_{x+2}\tilde g+
\big(\eta(x-2)-\eta(x-1)\big)\tau_{x}\tilde g\,.
\end{equation}
Formula \eqref{eln} gives the most general positive solution to \eqref{pes} while instead this is not the case for the rates of
Section \ref{odkd} (note for example that selecting $s=0$ in \eqref{eln} we obtain rates such that $\min\left\{\tilde{c}_{x+1,x},\tilde{c}_{x,x+1}\right\}=0$ while this is not always possible for the rates in Section \ref{odkd}).
This means that any model constructed with cycles like in Section \ref{odkd} can be obtained by \eqref{eln} for some $\tilde g$ and $s$;
however among the rates defined by \eqref{eln} there are some models for which the typical current cannot be decomposed
by cycles of length two and by cycles like in Figure 3.

We stress again that in this case, as well as in the followings, we obtain very general families
of models parametrized by arbitrary functions. Considering simple cylindric functions $\tilde g$ and $s$ it is
possible to obtain simple and completely explicit models.

\subsection{One-dimensional Glauber dynamics}

The stationary condition for a Glauber dynamics on the one dimensional torus is
\begin{equation}\label{stazb}
\mu(\eta)\sum_{x\in \mathbb Z_N}\left[c_x(\eta)-\frac{\mu(\eta^x)}{\mu(\eta)}c_x(\eta^x)\right]=0\,.
\end{equation}
We consider translational covariant rates that can be written as
\begin{equation}\label{rate}
c_x(\eta)=\eta(x)\tau_xc^-(\eta)+(1-\eta(x))\tau_xc^+(\eta)\,,
\end{equation}
where $D\left(c^\pm\right)\cap\{0\}=\emptyset$. We consider the case of
Bernoulli invariant measures of parameter $p$. The stationary condition \eqref{stazb} is equivalent to
\begin{eqnarray}\label{ste1}
& &\sum_{x\in \mathbb Z_N}\left(\eta(x)\tau_xc^-(\eta)+(1-\eta(x))\tau_xc^+(\eta)\right)\nonumber \\
& &=\sum_{x\in \mathbb Z_N}\left(\eta(x)\frac{(1-p)}{p}\tau_xc^+(\eta)+(1-\eta(x))\frac{p}{(1-p)}\tau_xc^-(\eta)\right)=0\,,
\end{eqnarray}
that with some algebra becomes
\begin{equation}\label{step3}
\sum_{x\in \mathbb Z_N}\left[\left(1-\frac{\eta(x)}{p}\right)\left(\tau_xc^+(\eta)-\frac{p}{(1-p)}\tau_xc^-(\eta)\right)\right]=0\,.
\end{equation}
As before we can naturally interpret \eqref{step3} as the orthogonality condition $\langle\phi_\eta,\mathbb I\rangle=0$ where the translational covariant vector field $\phi_\eta$ is defined by
\begin{equation}\label{deffi}
\phi_\eta(x,x+1):=\left(1-\frac{\eta(x)}{p}\right)\left(\tau_xc^+(\eta)-\frac{p}{(1-p)}\tau_xc^-(\eta)\right)\,,
\end{equation}
setting $\phi_\eta(x+1,x)=-\phi_\eta(x,x+1)$ by antisymmetry. Since we are in 1 dimension we have that \eqref{step3} holds
if and only if there exists a function $h$ such that
$$
\phi_\eta(x,x+1)=\tau_{x+1}h(\eta)-\tau_x h(\eta)\,.
$$
By translational covariance this relation is equivalent to
\begin{equation}\label{unasola}
c^+(\eta)-\frac{p}{(1-p)}c^-(\eta)=\frac{\tau_{1}h(\eta)- h(\eta)}{\left(1-\frac{\eta(0)}{p}\right)}\,.
\end{equation}
An important fact to observe is that the left hand side of \eqref{unasola} does not depend on $\eta(0)$ and then this must be true also for the right hand side. We have a general family of functions satisfying this constraint that is given by
\begin{equation}\label{genht}
h(\eta)= \left(1-\frac{\eta(-1)}{p}\right)\left(1-\frac{\eta(0)}{p}\right)\tilde{h}(\eta)\,,
\end{equation}
where $\tilde h$ is an arbitrary function such that $D(\tilde h)\cap \{-1,0\}=\emptyset$.
The other constraint that has to be satisfied is that the rates are nonnegative functions and this is obtained considering positive solutions of \eqref{unasola} by using \eqref{posol}. We obtain
\begin{equation}\label{formratesg1}
\left\{
\begin{array}{l}
c^+(\eta)=\left[\left(1-\frac{\eta(1)}{p}\right)\tau_1\tilde h-\left(1-\frac{\eta(-1)}{p}\right)\tilde h\right]_++s(\eta)\,,\\
c^-(\eta)=\frac{1-p}{p}\left(\left[\left(1-\frac{\eta(-1)}{p}\right)\tilde h-\left(1-\frac{\eta(1)}{p}\right)\tau_1\tilde h\right]_++s(\eta)\right)\,,
\end{array}
\right.
\end{equation}
where $s$ is an arbitrary non negative function such that $D(s)\cap \{0\}=\emptyset$.
Again the part with $s$ is the reversible contribution.

\subsection{Two dimensional Kawasaki dynamics }

In two dimensions the stationary condition under the hypothesis of Bernoulli invariant measure is analogous to \eqref{mauf}. We have indeed
\begin{equation}\label{dentino}
\sum_{x\in \mathbb Z_N^2}\sum_{i=1,2}\left[\left(\eta\left(x+e^{(i)}\right)-\eta(x)\right)\left(\tilde{c}_{x+e^{(i)},x}(\eta)-\tilde{c}_{x,x+e^{(i)}}
(\eta)\right)\right]=0\,,
\end{equation}
that can be interpreted as $\langle\gamma_\eta,\mathbb I\rangle=0$ where the translational covariant vector field $\gamma_\eta$ is given by \begin{equation}\label{kkk}
\gamma_\eta\left(x,x+e^{(i)}\right)=\left(\eta\left(x+e^{(i)}\right)-\eta(x)\right)\left(\tilde{c}_{x+e^{(i)},x}(\eta)
-\tilde{c}_{x,x+e^{(i)}}(\eta)\right)\,, \qquad i=1,2\,.
\end{equation}
The values of $\gamma_\eta\left(x+e^{(i)},x\right)$ are fixed by antisymmetry. The orthogonality condition is satisfied if and only if
for any fixed $\eta$ the vector $\gamma_\eta$ is obtained as a linear combination of a vector field in $\nabla \Lambda^0$,
a vector in $\delta\Lambda^2$ and the vector $\phi^{(1)}-\phi^{(2)}$. Since the vector filed $\gamma_\eta$ is translational covariant
there exist functions $h,g$ and a real number $\lambda$ such that
\begin{equation}\label{hodgefun22}
\gamma_\eta(e)=\big[\tau_{e^+}h(\eta)-\tau_{e^-}h(\eta)\big]+\big[\tau_{\mathfrak f^+(e)}g(\eta)-\tau_{\mathfrak f^-(e)}g(\eta)\big]\pm \lambda\,,
\end{equation}
where the sign $\pm$ has to be fixed as $+$ if $e=(x,x+e^{(1)})$  or $e=(x,x-e^{(2)})$ for some $x$ and has to be fixed as $-$
in the remaining cases.
We discuss two cases. In both cases the orthogonality with the vector $\mathbb I$ is verified but the splitting \eqref{hodgefun22} is not trivial.

\smallskip

The first case is as follows. Let $k(\eta)$ and $v(\eta)$ be two function having a structure like in Section \ref{1dkb} but respectively along the two directions of the plane. More precisely let $\tilde k$ be a function
such that $D\left(\tilde k\right)\subseteq \{je^{(1)}\,,\,j=1,\dots ,N\}$ and $D\left(\tilde k\right)\cap \{-2e^{(1)}, -e^{(1)}, 0, e^{(1)}\}=\emptyset$. Define then
$$
k(\eta):=\left[\eta\left(-e^{(1)}\right)-\eta\left(-2e^{(1)}\right)\right]
\left[\eta\left(e^{(1)}\right)-\eta\left(0\right)\right]\tilde k(\eta)\,.
$$
Likewise let $\tilde v$ be a function
such that $D\left(\tilde v\right)\subseteq \{je^{(2)}\,,\,j=1,\dots ,N\}$ and $D\left(\tilde v\right)\cap \{-2e^{(2)}, e^{(2)}, 0, e^{(2)}\}=\emptyset$. Define then
$$
v(\eta):=\left[\eta\left(-e^{(2)}\right)-\eta\left(-2e^{(2)}\right)\right]
\left[\eta\left(e^{(2)}\right)-\eta\left(0\right)\right]\tilde v(\eta)\,.
$$
We define $h^1(\eta):=k(\eta)+\tau_{e^{(1)}}k(\eta)$ and  $h^2(\eta):=v(\eta)+\tau_{e^{(2)}}v(\eta)$.
Finally we define
\begin{equation}\label{terr}
\left\{
\begin{array}{l}
\gamma_\eta(x,x+e^{(1)})=\tau_{x+e^{(1)}}h^1-\tau_x h^1\,,\\
\gamma_\eta(x,x+e^{(2)})=\tau_{x+e^{(2)}}h^2-\tau_x h^2\,.
\end{array}
\right.
\end{equation}
The orthogonality $\langle\gamma_\eta,\mathbb I\rangle=0$ follows by
$$
\langle\gamma_\eta,\phi^{(1)}\rangle=\langle\gamma_\eta,\phi^{(2)}\rangle=0\,,
$$
obtained like in the one dimensional case. Inserting \eqref{terr} in the left hand side of \eqref{kkk} we obtain
the rates $\tilde c$ like in Section \ref{1dkb}.

\smallskip

The second case that it is possible to describe is the following. Let $b$ be a function such that $D(b)\cap \left\{0,e^{(1)},e^{(2)},e^{(1)}+e^{(2)}\right\}=\emptyset$. Consider also $c^{(i)}$, $i=1,2$ two arbitrary constants. We define
\begin{equation}\label{clml}
\tilde{c}_{x,y}(\eta)
-\tilde{c}_{y,x}(\eta)=\left[\tau_{\mathfrak f^+(x,y)}b(\eta)-\tau_{\mathfrak f^-(x,y)}b(\eta)\right]\pm c^{(i)}\,,
\end{equation}
when $y=x\pm e^{(i)}$.
By the property of $D(b)$ it is possible to obtain by \eqref{clml} the rates $\tilde c$ using again \eqref{posol}.
According to \eqref{clml} if we define the vector field $\psi_\eta(x,y):=\tilde{c}_{x,y}(\eta)
-\tilde{c}_{y,x}(\eta)$, we have $\psi_\eta\in\delta\Lambda^2\oplus\Lambda^1_H$ for any $\eta$. We define also the vector field
$\phi_\eta(x,y):=\eta(y)-\eta(x)$ and for any $\eta$ we have $\phi_\eta\in \nabla \Lambda^0$. Since the stationary conditions
\eqref{dentino} can also naturally be interpreted as $\langle\phi_\eta,\psi_\eta\rangle=0$ and the two discrete vector fields belong to orthogonal subspaces the stationarity conditions are automatically satisfied.

\subsection{Two dimensional Glauber dynamics}
We consider now the two dimensional Glauber dynamics but the same construction can be done also in higher dimensions. The stationary condition for a Bernoulli invariant measure of parameter $p$ can be written again like \eqref{step3}
\begin{equation}\label{step32}
\sum_{x\in \mathbb Z_N^2}\left[\left(1-\frac{\eta(x)}{p}\right)\left(\tau_xc^+(\eta)-\frac{p}{(1-p)}\tau_xc^-(\eta)\right)\right]=0\,.
\end{equation}
This expression is of the form $\sum_{x\in \mathbb Z_N^2}\tau_x g(\eta)=0$ where the function $g$ is given by
\begin{equation}\label{carsoli}
g(\eta)=\left(1-\frac{\eta(0)}{p}\right)\left(c^+(\eta)-\frac{p}{(1-p)}c^-(\eta)\right)\,.
\end{equation}
Given a $\tilde g\in \Lambda^0$ such that $\sum_{x\in \mathbb Z_N^2}\tilde g(x)=0$ then there exists a $\phi\in \Lambda^1$ such that
$\tilde g(x)=\nabla\cdot \phi(x)$. This fact can be proved in several ways, for example considering a $\phi$ of gradient type
$\phi(x,y)=k(y)-k(x)$ we have that $k$ has to satisfy the discrete Poisson equation $\Delta k=\tilde g$ that has always a solution.
In our case we have that $\tilde g(x)=\tilde g_\eta(x)$ depends on $\eta$ so that
there exists an $\eta$ dependent discrete vector field $\phi_\eta$ such that $\nabla\cdot \phi_\eta=\tilde g_\eta$. Since $\tilde g_\eta(x)=\tau_xg(\eta)$ is translational covariant, with the same arguments as in Theorem \ref{belteo2} we can prove that
the $\eta$ dependent discrete vector field $\phi_\eta$ can be chosen translational covariant too. The most general translational covariant discrete vector field has the following structure. Let $h(\eta), v(\eta)$ be two function then we define
\begin{equation}\label{dvfc}
\left\{
\begin{array}{l}
\phi_\eta(x,x+e^{(1)})=\tau_x h(\eta)\,,  \\
\phi_\eta(x,x+e^{(2)})=\tau_x v(\eta)\,,
\end{array}
\right.
\end{equation}
and this is the most general translational covariant $\eta$ dependent discrete vector field.
We obtain that the stationary condition is satisfied if and only if there exists two functions $h,v$ such that
\begin{equation}\label{quest}
g(\eta)= \nabla\cdot \phi_\eta=h(\eta)-\tau_{-e^{(1)}}h(\eta)+v(\eta)-\tau_{-e^{(2)}}v(\eta)\,,
\end{equation}
where $g$ is given by \eqref{carsoli}.
As in the one dimensional case to solve the equation \eqref{quest} we have to fix the functions $h,v$ in such a way that
$$
\frac{h(\eta)-\tau_{-e^{(1)}}h(\eta)+v(\eta)-\tau_{-e^{(2)}}v(\eta)}{\left(1-\frac{\eta(0)}{p}\right)}
$$
has a domain $D$ such that $D\cap \{0\}=\emptyset$. A general family that satisfies this constraint is given by
$$
\left\{
\begin{array}{l}
h(\eta)=\left(1-\frac{\eta(0)}{p}\right)\left(1-\frac{\eta\left(e^{(1)}\right)}{p}\right)\tilde h(\eta)\,,\\
v(\eta)=\left(1-\frac{\eta(0)}{p}\right)\left(1-\frac{\eta\left(e^{(2)}\right)}{p}\right)\tilde v(\eta)\,,
\end{array}
\right.
$$
where $D\left(\tilde h\right)\cap \{0,e^{(1)}\}=\emptyset$ and $D\left(\tilde v\right)\cap \{0,e^{(2)}\}=\emptyset$.
Under these assumptions we can find the positive solutions $c^\pm$ using again \eqref{posol}. Writing only the non reversible contribution we have that $c^+(\eta)$ is equal to
\begin{align}\label{uuu}
&\left[\left(1-\frac{\eta\left(e^{(1)}\right)}{p}\right)\tilde h-\left(1-\frac{\eta\left(-e^{(1)}\right)}{p}\right)\tau_{-e^{(1)}}\tilde h\right.\nonumber \\
&\left.+\left(1-\frac{\eta\left(e^{(2)}\right)}{p}\right)\tilde v-\left(1-\frac{\eta\left(-e^{(2)}\right)}{p}\right)\tau_{-e^{(2)}}\tilde v\right]_+
\end{align}
while $c^-(\eta)$ is given by $c^-(\eta)=\frac{1-p}{p}\left[-i(\eta)\right]_+$, where $i(\eta)$ is the function appearing inside the positive part in \eqref{uuu}.

\section*{Acknowledgments}
D. G. wrote part of this work during his stay at the Institut Henri  Poincare - Centre Emile Borel during the trimester « Stochastic Dynamics Out  of Equilibrium » and thanks this institution for hospitality and support.


\begin{thebibliography}{99}

\bibitem{AKM} C. Arita, P. L. Krapivsky, K. Mallick \emph{Variational calculation of transport coefficients in diffusive lattice gases}
 arXiv:1611.07719

\bibitem{GO} J. Bang-Jensen,  G. Gutin {\emph Digraphs. Theory, algorithms and applications} Sprint sinceger Monographs in Mathematics. Springer-Verlag London, Ltd., London, 2001

\bibitem{BBC}  J. Barr\'{e}, C. Bernardin, R. Chetrite
    \emph{Density large deviations for multidimensional stochastic hyperbolic conservation laws}
arXiv:1702.03769    


\bibitem{BDGJLstoc-int} L. Bertini , A. De Sole, D. Gabrielli, G. Jona-Lasinio, C. Landim  \emph{Stochastic interacting particle systems out of equilibrium} J. Stat. Mech. , P07014 (2007)




\bibitem{BFG} L. Bertini,  A. Faggionato,  D. Gabrielli,  \emph{Large deviations of the empirical flow for continuous time Markov chains} Ann. Inst. Henri Poincar\'{e} Probab. Stat. {\textbf 51} (2015), no. 3, 867–900.

\bibitem {B}  N. Biggs  \emph{Algebraic graph theory} Second edition. Cambridge Mathematical Library. Cambridge University Press, Cambridge, 1993.

\bibitem{Bi} J. Bierkens \emph{Non-reversible Metropolis-Hastings} Statistics and Computing pp. 1--16
(2015)


\bibitem{BB} A. Borodin, A. Bufetov \emph{An irreversible local Markov chain that preserves the six vertex model on a torus}
 arXiv:1509.05070


\bibitem{CT}  I. Corwin, F.L. Toninelli \emph{Stationary measure of the driven two-dimensional q-Whittaker particle system on the torus} Electron. Commun. Probab. {\bf 21} (2016), Paper No. 44,





\bibitem{D}  P. Diaconis   \emph{The Markov chain Monte Carlo revolution} Bull. Amer. Math. Soc. (N.S.) {\bf 46}
(2009), no. 2, 179–205.

\bibitem{FGS} L. Fajfrov\`{a}, T. Gobron, E. Saada  \emph{Invariant measures of Mass
Migration Processes} Electron. J. Probab., {\bf 21}, no. 60 (2016)
1–52.


\bibitem{FW} Freidlin M.I., Wentzell A.D. \emph{Random perturbations of dynamical systems} Third edition, Grundlehren der Mathematischen Wissenschaften {\bf 260} Springer, Heidelberg, 2012.


\bibitem{GJLV} D. Gabrielli , G. Jona-Lasinio , C. Landim , M.E. Vares
\emph{Microscopic reversibility and thermodynamic fluctuations} Boltzmann's legacy 150 years after his birth (Rome, 1994), 79–87,
Atti Convegni Lincei, {\textbf 131}, Accad. Naz. Lincei, Rome, (1997)

\bibitem{GK} D. Gabrielli, P. L. Krapivsky in preparation


\bibitem{GV} D. Gabrielli, C. Valente  {\emph Which random walks are cyclic?}
ALEA, Lat. Am. J. Probab. Math. Stat. {\bf 9}, 231-267 (2012)


\bibitem{G}  C. Godr\`{e}che \emph{Rates for irreversible Gibbsian Ising models}
J. Stat. Mech. Theory Exp. 2013, no. 5, P05011

\bibitem{GL}  C. Godr\`{e}che, J.M. Luck \emph{Single-spin-flip dynamics of the Ising chain} J. Stat. Mech. Theory Exp. 2015, no. 5, P05033

\bibitem{KJZ} M. Kaiser, R.L. Jack, J. Zimmer \emph{Acceleration of convergence to equilibrium in Markov chains by breaking detailed balance} arXiv:1611.06509

\bibitem{LG}  J.M. Luck, C. Godr\`{e}che  \emph{Nonequilibrium stationary states with Gibbs measure for two or three species of interacting particles} J. Stat. Mech. Theory Exp. 2006, no. 8, P08009




\bibitem{Li}  T.M. Liggett \emph{Interacting particle systems} Grundlehren der Mathematischen Wissenschaften {\textbf 276} Springer-Verlag, New York, 1985.

\bibitem{L} L. Lov\'{a}sz
\emph{Discrete analytic functions: an exposition}  Surveys in differential geometry. Vol. IX, 241–273,
Surv. Differ. Geom., {\textbf 9}, Int. Press, Somerville, MA, (2004).

\bibitem{Metr} D.P. Landau, K. Binder
\emph{A guide to Monte Carlo simulations in statistical physics}
Fourth edition. Cambridge University Press, Cambridge, (2015).


\bibitem{KL99}
    C. Kipnis and C. Landim, {\it Scaling Limits of Interacting Particle Systems} (Springer, New York,  1999).



\bibitem{MC} J. MacQueen  \emph{Circuit processes} Ann. Probab. {\bf 9}, 604--610 (1981).

\bibitem{N} Y. Nagahata \emph{The gradient condition for one-dimensional
symmetric exclusion processes} J.  Statist. Phys. \textbf{ 91}, No. 3/4. 587-602, (1998)


\bibitem{PSS} A. Procacci, B. Scoppola, E. Scoppola \emph{Effects of boundary conditions on irreversible dynamics}
arXiv:1703.04511


\bibitem{RBS} L. Rey-Bellet, K. Spiliopoulos \emph{Improving the convergence of reversible samplers}
J. Stat. Phys. {\bf 164} (2016), no. 3, 472–494.

\bibitem{SB} R. D. Schrama, G.T. Barkema \emph{Monte Carlo methods beyond detailed balance}
Physica A {\bf 418} (2015) 88–93.



\bibitem{Spohn}
    H. Spohn, \emph{Large Scale Dynamics of Interacting Particles}
    (Springer-Verlag, New York, 1991).

\bibitem{VY}  S.R.S. Varadhan,  H.T. Yau \emph{Diffusive limit of lattice gas with mixing conditions} Asian J. Math. {\textbf 1} (1997), no. 4, 623–678.


\end{thebibliography}
\end{document}